\documentclass[11pt]{article}
\usepackage{graphicx,latexsym}
\usepackage{amsmath}
\usepackage{epstopdf}
\usepackage{colortbl}
\usepackage{arydshln}
\usepackage{makecell}
\usepackage{color}

\newcommand{\avec}{{\bf a}}

\newcommand{\xvec}{{\bf x}}

\newcommand{\yvec}{{\bf y}}

\newcommand{\lnon}{\overline}

\newtheorem{proposition}{Proposition}

\newtheorem{theorem}{Theorem}
\newtheorem{remark}{Remark}
\newtheorem{example}{Example}

\newtheorem{definition}{Definition}

\newenvironment{proof}{{\noindent \em Proof.}}{\hfill $\fbox{}$ \vspace*{5mm}}

\begin{document}

\title{On the Number of Observation Nodes in Boolean Networks}

\author{Liangjie Sun, Wai-Ki Ching, Tatsuya Akutsu}
\date{}

\maketitle
\begin{abstract}
A Boolean network (BN) is called observable if any initial state can be uniquely determined from the output sequence. In the existing literature on observability of BNs, there is almost no research on the relationship between the number of observation nodes and the observability of BNs, which is an important and practical issue. In this paper, we mainly focus on three types of BNs with $n$ nodes (i.e., $K$-AND-OR-BNs, $K$-XOR-BNs, and $K$-NC-BNs, where $K$ is the number of input nodes for each node
and NC means nested canalyzing) and study the upper and lower bounds of the number of observation nodes for these BNs. First, we develop a novel technique using information entropy to derive a general lower bound of the number of observation nodes, and conclude that the number of observation nodes cannot be smaller than $\left[(1-K)+\frac{2^{K}-1}{2^{K}}\log_{2}(2^{K}-1)\right]n$ to ensure that any $K$-AND-OR-BN is observable, and similarly, some lower bound is also obtained for $K$-NC-BNs. Then for any type of BN, we also develop two new techniques to infer the general lower bounds, using counting identical states at time 1 and counting the number of fixed points, respectively. On the other hand, we derive nontrivial upper bounds of the number of observation nodes by combinatorial analysis of several types of BNs. Specifically, we indicate that $\left(\frac{2^{K}-K-1}{2^{K}-1}\right)n,~1$, and $\lceil \frac{n}{K}\rceil$ are the best case upper bounds for $K$-AND-OR-BNs, $K$-XOR-BNs, and $K$-NC-BN, respectively.
\end{abstract}

\section{Introduction}
The Boolean network (BN), proposed by Kauffman \cite{kauffman1969metabolic}, is a discrete mathematical model for describing and analyzing gene regulatory networks. In this model, the state of a node can be expressed by a binary variable: $1$ (active)
or $0$ (inactive), at each time step. Moreover, the regulatory interactions between nodes can be described by Boolean functions, which determine the state of a node at the next time instant through the state of other nodes. It is worth mentioning that there is a special type of Boolean function called canalyzing function, which is characterized by the presence of at least one input such that the value of the function is determined for one of the two values of that input, regardless of the values of the other inputs.  The nested canalyzing function is a subclass of canalyzing functions. Notably, from the existing literature \cite{Harris2002,Kauffman2004}, many of the biologically relevant Boolean functions seem to belong to the canalyzing function class. Furthermore, Boolean control networks (BCNs) were proposed in \cite{Ideker} by taking the influence of environmental factors on gene regulatory networks as control inputs. To date, a range of meaningful results on BNs and BCNs have emerged, see for instance, \cite{li2024,lin2024,wang2024,yu2024}.

The property of a BN that enables to distinguish the entire system initial state from its outputs $\yvec(0),\ldots,\yvec(N)$ for some $N>0$, where $\yvec(t)=[y_{1}(t),\ldots,y_{m}(t)]$, is called observability. We say that $y_{1}(t),\ldots,y_{m}(t)$ are the observation nodes and $m$ is the number of observation nodes. Recently, many observability problems have been studied \cite{lin2024,Laschov2013,wang2024,Zhou2019,yu2020,zhang2022,zhang2023}. Among these problems, the minimal observability problem is particularly striking, which is how to add a minimal number of measurements to make an unobservable BN observable. In \cite{Weiss2018}, a necessary and sufficient condition for observability of conjunctive BNs was derived by using the dependency graph, and then an algorithm for solving the minimal observability problem for a conjunctive BN was proposed. Later, they used a similar graph-theoretic approach to derive a sufficient (but not necessary) condition for observability of general BCNs in \cite{Weiss20182}. Based on the sufficient condition for observability, two sub-optimal algorithms for solving the minimal observability problem for a general BCNs were proposed. On the other hand, the observability decomposition problem has also received extensive attention, which is whether a BN can be divided into an observable subsystem and an unobservable subsystem. In \cite{li2022TAC}, a necessary and sufficient vertex partition condition for the observability decomposition was obtained, and then an algorithm for checking the realizability of observability decomposition of a BCN was proposed. This algorithm not only separates the unobservable part from the considered BCN as much as possible, but also ensures that the remaining part is observable. After that, in \cite{li2024TAC}, the realizability of observability decomposition of BCNs under different definitions of observability was studied based on the observability criteria and some matrix discriminants. Furthermore, the strongly structurally observable graphs were purposed in \cite{zhu2024} as one of fundamental characterizations of observable discrete-time iterative systems, particularly for Boolean networks, from the perspective of network structures. Moreover, the minimal node control problem was solved in a polynomial amount of time to establish the explicit relationship between control nodes to observable paths.

In general, the number of observation nodes plays an important role in the observability of BNs, but to our knowledge, there is no theoretical study on the relationship between the number of observation nodes and the observability of BNs. In this paper, we mainly focus on three types of BNs: (i) $K$-AND-OR-BNs where each Boolean function is either AND or OR of $K$ literals, (ii) $K$-XOR-BNs where each Boolean function includes XOR operations only of $K$ literals, and (iii) $K$-NC-BNs where each Boolean function is a nested canalyzing function of $K$ literals, and we study the upper and lower bounds of the number of observation nodes for these BNs. To be more specific, in this paper, we consider the bounds shown in Table \ref{table1}, where BNs are restricted to some specified class, and $b$ denotes the corresponding bound.
\begin{table}[ht!]
  \centering
  \caption{The meaning of each bound.}
  \resizebox{\linewidth}{!}{
  \label{table1}
  \begin{tabular}{|c|p{12cm}|}
   \hline
   General lower bound & For any BN, the number of observation nodes is not smaller than $b$. \\ \hline
   Best case upper bound & For some BN, the minimum number of observation nodes is at most $b$, which is currently the smallest. \\ \hline
   Worst case lower bound & For some BN, the number of observation nodes is not smaller than $b$, which is currently the largest. \\ \hline
   General upper bound & For any BN, the minimum number of observation nodes is at most $b$. \\ \hline
  \end{tabular}}
\end{table}
Note that in general, the following relation holds: general lower bound $\leq$ best case upper bound $<$ worst case lower
bound $\leq$ general upper bound. If general lower bound $=$ best case upper bound (resp., worst case
lower bound $=$ general upper bound), these bounds are called tight or optimal. Since $n$ is a trivial general upper bound, the worst case lower bound is tight if it is $n$.

The results are summarized in Table \ref{table2}, where $COUNT(\avec^{j})$ and $l$ represent the total number of occurrences of $\avec^{i}$ in $\{\xvec(1;\xvec(0)),\xvec(0)\in\{0,1\}^{n}\}$ and the number of fixed points in a BN, respectively, $K$ and $m$ denote the number of input nodes and observation nodes, respectively, $\beta_{K}=-\left[\frac{2^{K-1}-1}{2^K}\log_{2}\left(\frac{2^{K-1}-1}{2^{K}}\right)+\frac{2^{K-1}+1}{2^K}\log_{2}\left(\frac{2^{K-1}+1}{2^{K}}\right)\right]$, and $n$ is the number of nodes in a BN.

First, for $K$-AND-OR-BNs, a novel technique for deriving a general lower bound of the number of observation nodes using information entropy of BNs is developed. Here, for 2-AND-OR-BNs, the general lower bound $0.188n$ is obtained in Theorem \ref{theorem1}. Then we extend the above result to $K$-AND-OR-BNs ($K>2$) and find that $m\geq \left[(1-K)+\frac{2^{K}-1}{2^{K}}\log_{2}(2^{K}-1)\right]n$ must hold so that a $K$-AND-OR-BN ($K>2$) is observable in Theorem \ref{theorem2}. Based on Proposition \ref{proposition1} and Proposition \ref{proposition5}, we can see that the worst case lower bound for $K$-AND-OR-BNs is always $n$. Later, we also present a new technique by counting identical states $\xvec(1;\xvec(0))$ to study the general lower bound of the number of observation nodes, which is applicable to any type of BN, and we show that $\log_{2}[\max_{\avec^{j}}COUNT(\avec^{j})]$ is the general lower bound for any type of BN in Proposition \ref{proposition4}. Moreover, by counting the number of fixed points, a general lower bound $\log_{2}l$ for any type of BN is derived in Proposition \ref{proposition9}.
At last, in Proposition \ref{proposition}, we show that for any $K$-NC-BN, the number of observation nodes is not smaller than $(1-\beta_{K})n$, where $\beta_{K}=-\left[\frac{2^{K-1}-1}{2^K}\log_{2}\left(\frac{2^{K-1}-1}{2^{K}}\right)+\frac{2^{K-1}+1}{2^K}\log_{2}\left(\frac{2^{K-1}+1}{2^{K}}\right)\right]$.

On the other hand, we derive nontrivial upper bounds of the number of observation nodes by combinatorial analysis of several types of BNs. Concretely speaking, in Proposition \ref{proposition2} and Theorem \ref{theorem}, we infer that $\left(\frac{2^{K}-K-1}{2^{K}-1}\right)n$ is the best case upper bound for $K$-AND-OR-BNs. For $K$-XOR-BNs, we notice that there exists a $K$-XOR-BN that is observable with $m=K$, where $n=K+1$ in Proposition \ref{proposition8}, and that there exists a $K$-XOR-BN that is observable with $m=\frac{K}{K+1}n$, where $K$ is odd and $(n\mod K+1)=0$ in Proposition \ref{proposition10}. In addition, the best case upper bound $m=1$ for $K$-XOR-BN is shown in Proposition \ref{proposition11}. Finally, in Theorem \ref{theorem5}, we show that there exists a specific $K$-NC-BN, which is observable with $m=\lceil \frac{n}{K}\rceil$.
Here, we find that the upper bound for $K$-NC-BNs can break the lower bound for $K$-AND-OR-BNs, which implies that there exist $K$-NC-BNs that are easier to
observe than any other $K$-AND-OR-BNs.

\begin{table}[ht!]
  \centering
  \caption{Summary of results.}
\label{table2}
\resizebox{\linewidth}{!}{
\begin{tabular}{|l|c|c|c|}
   \hline
  & type & observation nodes ($m$)  & characteristic\\ \hline
  Theorem \ref{theorem1} & 2-AND-OR-BNs & $0.188n$ & General lower bound using information entropy of BNs\\ \hline
  Proposition \ref{proposition1} & 2-AND-OR-BNs & $n$ & Worst case lower bound\\ \hline
  Proposition \ref{proposition2} & 2-AND-OR-BNs & $\frac{n}{3}$ & Best case upper bound, where the general upper bound is always $n$ \\ \hline
  Proposition \ref{proposition4} & Any BNs & $\log_{2}[\max_{\avec^{j}}COUNT(\avec^{j})]$ & General lower bound by counting identical states $\xvec(1;\xvec(0))$ \\ \hline
  Theorem \ref{theorem2} & $K$-AND-OR-BNs,~$K>2$ & $[(1-K)+\frac{2^{K}-1}{2^{K}}\log_{2}(2^{K}-1)]n$ & General lower bound using information entropy of BNs  \\ \hline
  Proposition \ref{proposition5} & $K$-AND-OR-BNs,~$K>2$ & $n$ & Worst case lower bound\\ \hline
  Theorem \ref{theorem} & $K$-AND-OR-BNs,~$K>2$ & $\left(\frac{2^{K}-K-1}{2^{K}-1}\right)n$ & Best case upper bound, where the general upper bound is always $n$ \\ \hline
  Proposition \ref{proposition7} & $2$-XOR-BNs & $1$ & Best case upper bound, where the general upper bound is always $n$ \\ \hline
  Remark \ref{remark4} & $2$-XOR-BNs & $n$ & Worst case lower bound \\ \hline
  Proposition \ref{proposition8} & $K$-XOR-BNs,~$K>2,n=K+1$ & $K$ & Best case upper bound \\ \hline
  Proposition \ref{proposition9} & Any BNs & $\log_{2}l$ & General lower bound by counting the number of fixed points \\ \hline
  Proposition \ref{proposition10} & $K$-XOR-BNs,~$K>2$ & $\frac{K}{K+1}n$ & Best case upper bound \\ \hline
  Proposition \ref{proposition11} & $K$-XOR-BNs,~$K>2$ & $1$ & Best case upper bound, where the general upper bound is always $n$ \\ \hline
  Theorem \ref{theorem5} & $K$-NC-BNs,~$K>2$ & $\lceil \frac{n}{K}\rceil$ & Best case upper bound, where the general upper bound is always $n$ \\ \hline
  Proposition \ref{proposition} & $K$-NC-BNs,~$K>2$ & $(1-\beta_{K})n$ & General lower bound using information entropy of BNs \\ \hline
\end{tabular}}
\end{table}

\section{Preliminaries}
For two integers $a$ and $b$ with $a<b$, let $[a,b]:=\{a,a+1,\ldots,b\}$.

Consider the following synchronous BN:
\begin{eqnarray*}
x_{i}(t+1) & = & f_{i}(x_{i_{1}}(t),\ldots,x_{i_{k_{i}}}(t)),\quad \forall i\in[1,n],\\
y_{j}(t) & = & x_{\phi(j)}(t),\quad \forall j\in[1,m],
\end{eqnarray*}
where every $x_{i}(t)$ and $y_{j}(t)$ takes values in $\{0,1\}$, $f_{i}:\{0,1\}^{k_{i}}\rightarrow\{0,1\}$ is a Boolean function and $\phi(j)$ is a function from $[1,m]$ to $[1,n]$.\footnote{$y_{j}(t)$ can be determined as $y_{j}(t)=g_{j}(x_{1}(t),\ldots,x_{n}(t))$.} Here, we say that $x_{i_{1}},\ldots,x_{i_{k_{i}}}\in\{x_{1},\ldots,x_{n}\}$ is relevant, which means that for each $i_{j}\in\{i_{1},\ldots,i_{k_{i}}\}$, there exists $[b_{1},\ldots,b_{i_{k_{i}}}]\in\{0,1\}^{k_{i}}$ such that  $f_{i}(b_{i_{1}},\ldots,b_{i_{j}},\ldots,b_{i_{k_{i}}})\neq f_{i}(b_{i_{1}},\ldots,\overline{b_{i_{j}}},\ldots,b_{i_{k_{i}}})$, where $\overline{b_{i_{j}}}$ means the negation (NOT) of $b_{i_{j}}$.

Denote the state of a BN at time $t$ by $\xvec(t):=[x_{1}(t),\ldots,x_{n}(t)]$ and the output by $\yvec(t):=[y_{1}(t),\ldots,y_{m}(t)]$. The state of a BN with initial state $\xvec(0)$ at time $t$, is expressed as $\xvec(t;\xvec(0))$, and the corresponding output is expressed as $\yvec(t;\xvec(0))$. Moreover, $\xvec(t;\xvec(0),x_{i}(t)=j)$ represents the state of a BN with initial state $\xvec(0)$ and $x_{i}(t)=j$ at time $t$.

\begin{definition}
A BN is observable on $[0,N]$ if any two different initial states $\xvec^{1}(0),\xvec^{2}(0)\in\{0,1\}^{n}$ are distinguishable on the time interval $[0,N]$.
\end{definition}
In other words, a BN is observable on $[0,N]$, if any two different initial states $\xvec^{1}(0)$ and $\xvec^{2}(0)$ yield different output sequences $\yvec(0;\xvec^{1}(0)),\ldots,\yvec(N;\xvec^{1}(0))$ and $\yvec(0;\xvec^{2}(0)),\ldots,\yvec(N;\xvec^{2}(0))$, which means that for given $\yvec(t)=[y_{1}(t),\ldots,y_{m}(t)],~t\in[0,N]$, the initial state $\xvec(0)$ can always be uniquely determined. We say that $y_{1}(t),\ldots,y_{m}(t)$ are the observation nodes and $m$ is the number of observation nodes. We also say that a BN is observable if the BN is observable on $[0,N]$ for some $N$.

Assuming there are $r$ different states $\xvec(1;\xvec(0))$ for all $2^{n}$ initial states $\xvec(0)$, and let $\avec^{1},\ldots,\avec^{r}$
be the different states, i.e., $\{\xvec(1;\xvec(0)):\xvec(0)\in\{0,1\}^{n}\}=\{\avec^{1},\ldots,\avec^{r}\}$. $COUNT(\avec^{j}:\xvec(1;\xvec(0)),\xvec(0)\in\{0,1\}^{n})$ represents the total number of occurrences of $\avec^{i}$ in $\{\xvec(1;\xvec(0)),\xvec(0)\in\{0,1\}^{n}\}$. Without any confusion, we remove the range $\{\xvec(1;\xvec(0)),\xvec(0)\in\{0,1\}^{n}\}$ and use $COUNT(\avec^{j})$. Here,
\begin{eqnarray*}
\sum_{j=1}^{r}COUNT(\avec^{j})=2^{n}.
\end{eqnarray*}

\begin{example}\label{example1}
Consider the following BN:
\begin{eqnarray*}
x_{1}(t+1) & = & x_{1}(t)\wedge x_{3}(t),\\
x_{2}(t+1) & = & \overline{x_{1}(t)}\wedge x_{3}(t),\\
x_{3}(t+1) & = & x_{1}(t)\wedge x_{2}(t).
\end{eqnarray*}
Then, we have the following state transition table (Table \ref{table3}).
\begin{table}[ht!]
  \centering
  \caption{State transition table in Example \ref{example1}.}
  \label{table3}
\begin{tabular}{|ccc|ccc|}
\hline
$x_{1}(0)$&$x_{2}(0)$&$x_{3}(0)$&$x_{1}(1)$&$x_{2}(1)$&$x_{3}(1)$\\
\hline
$0$&$0$&$0$&$0$&$0$&$0$\\
$0$&$0$&$1$&$0$&$1$&$0$\\
$0$&$1$&$0$&$0$&$0$&$0$\\
$0$&$1$&$1$&$0$&$1$&$0$\\
$1$&$0$&$0$&$0$&$0$&$0$\\
$1$&$0$&$1$&$1$&$0$&$0$\\
$1$&$1$&$0$&$0$&$0$&$1$\\
$1$&$1$&$1$&$1$&$0$&$1$\\
\hline
\end{tabular}
\end{table}
There are $5$ different states $\xvec(1;\xvec(0))$ for all initial states $\xvec(0)$, which are
\begin{eqnarray*}
[0,0,0],[1,0,0],[0,1,0],[0,0,1],[1,0,1].
\end{eqnarray*}
Here, $COUNT([0,0,0])=3,COUNT([1,0,0])=1,COUNT([0,1,0])=2,COUNT([0,0,1])=1,COUNT([1,0,1])=1$.
\end{example}

\begin{example}\label{example2}
Consider the following BN:
\begin{eqnarray*}
x_{1}(t+1) & = & \overline{x_{2}(t)}\wedge x_{3}(t),\\
x_{2}(t+1) & = & x_{2}(t)\wedge x_{3}(t),\\
x_{3}(t+1) & = & \overline{x_{2}(t)}\wedge\overline{x_{3}(t)},\\
x_{4}(t+1) & = & x_{1}(t)\wedge x_{4}(t).
\end{eqnarray*}
Then, we have the following state transition table (Table \ref{table4}).
\begin{table}[ht!]
  \centering
  \caption{State transition table in Example \ref{example2}.}
  \label{table4}
\begin{tabular}{|cccc|cccc|}
\hline
$x_{1}(0)$&$x_{2}(0)$&$x_{3}(0)$&$x_{4}(0)$&$x_{1}(1)$&$x_{2}(1)$&$x_{3}(1)$&$x_{4}(1)$\\
\hline
$0$&$0$&$0$&$0$&$0$&$0$&$1$&$0$\\
$0$&$0$&$0$&$1$&$0$&$0$&$1$&$0$\\
$0$&$0$&$1$&$0$&$1$&$0$&$0$&$0$\\
$0$&$0$&$1$&$1$&$1$&$0$&$0$&$0$\\
$0$&$1$&$0$&$0$&$0$&$0$&$0$&$0$\\
$0$&$1$&$0$&$1$&$0$&$0$&$0$&$0$\\
$0$&$1$&$1$&$0$&$0$&$1$&$0$&$0$\\
$0$&$1$&$1$&$1$&$0$&$1$&$0$&$0$\\
$1$&$0$&$0$&$0$&$0$&$0$&$1$&$0$\\
$1$&$0$&$0$&$1$&$0$&$0$&$1$&$1$\\
$1$&$0$&$1$&$0$&$1$&$0$&$0$&$0$\\
$1$&$0$&$1$&$1$&$1$&$0$&$0$&$1$\\
$1$&$1$&$0$&$0$&$0$&$0$&$0$&$0$\\
$1$&$1$&$0$&$1$&$0$&$0$&$0$&$1$\\
$1$&$1$&$1$&$0$&$0$&$1$&$0$&$0$\\
$1$&$1$&$1$&$1$&$0$&$1$&$0$&$1$\\
\hline
\end{tabular}
\end{table}
There are $8$ different states $\xvec(1;\xvec(0))$ for all initial states $\xvec(0)$, which are
\begin{eqnarray*}
[0,0,0,0],[1,0,0,0],[0,1,0,0],[0,0,1,0],[0,0,0,1],[1,0,0,1],[0,1,0,1],[0,0,1,1].
\end{eqnarray*}
Here, $COUNT([0,0,0,0])=COUNT([1,0,0,0])=COUNT([0,1,0,0])=COUNT([0,0,1,0])=3$ and
$COUNT([0,0,0,1])=COUNT([1,0,0,1])=COUNT([0,1,0,1])=COUNT([0,0,1,1])=1$.
\end{example}

A BN is called a $K$-AND-OR-BN if each Boolean function $f_{i}$ is either AND or OR of $K$ literals, and a BN is called a $K$-XOR-BN if every Boolean function $f_{i}$ includes XOR operations only of $K$ literals. Note that in a $K$-AND-OR-BN, some nodes can be AND nodes and some nodes
can be OR nodes. In addition, we refer to a BN as a $K$-NC-BN, if its Boolean function is a nested canalyzing function of $K$ literals. The nested canalyzing function, which is reported to be biologically significant, is defined as follows.

\begin{definition}\cite{akutsu11}
A Boolean function is nested canalyzing over $x_{1},\ldots,x_{n}$, if and only if it can be represented as
\begin{eqnarray*}
f=\ell_{1}\vee\cdots\vee\ell_{k_{1}-1}\vee(\ell_{k_{1}}\wedge\cdots\wedge\ell_{k_{2}-1}\wedge(\ell_{k_{2}}\vee\cdots\vee\ell_{k_{3}-1}\vee(\cdots))),
\end{eqnarray*}
where $\ell_{i}\in\{x_{1},\overline{x_{1}},x_{2},\overline{x_{2}},\ldots,x_{n},\overline{x_{n}}\}$ and $1\leq k_{1}<k_{2}<\cdots$.
\end{definition}

Then we recall the definition of information entropy.
\begin{definition}
Let $p_{i},~i\in[1,2^{n}]$ be the probability that the $i$-th event occurs in a given BN, where $0\leq p_{i}\leq1$ for all $i\in[1,2^{n}]$ and $\sum_{i=1}^{2^{n}}p_{i}=1$. Then, the entropy of the BN is given by $-\sum_{i=1}^{2^{n}}p_{i}\log_{2}p_{i}$, where $0\log_{2}0=0$.
\end{definition}
For more details on the information entropy of BNs, please refer to \cite{guo22}.

In this paper, we study the upper and lower bounds of $m$ so that any initial state $\xvec(0)$ can be determined from $\yvec(0),\ldots,\yvec(N)$ for some $N>0$ (i.e., the BN is observable), where $\yvec(t)=[y_{1}(t),\ldots,y_{m}(t)]$.
\section{Boolean networks consisting of AND or OR functions}
In this section, we study the upper and lower bounds of the number of observation nodes $m$ in BNs consisting of AND or OR functions.

For the simplicity, we first focus on the case where each $f_{i},~i\in[1,n]$ is AND or OR of two literals. We call these BNs 2-AND-OR-BNs.

Using information theoretic analysis on BNs \cite{guo22}, the following result is obtained.

\begin{theorem}\label{theorem1}
For any 2-AND-OR-BN, $m\geq0.188n$ must hold so that the BN is observable.
\end{theorem}
\begin{proof}
Since we need to discriminate all $2^{n}$ initial states $\xvec(0)$ from time series data of $\yvec$, we need the information quantity of $n$ bits from $[\yvec(0),\ldots,\yvec(N)]$.

Consider the case of $x_{1}(t+1)=x_{1}(t)\wedge x_{2}(t)$, where the same result should hold for any AND or OR functions of two literals. Then, the state transitions from $\xvec(0)$ to $x_{1}(1)$ are given as below (Table \ref{table5}).
\begin{table}[ht!]
  \centering
  \caption{State transition table under the case of $x_{1}(t+1)=x_{1}(t)\wedge x_{2}(t)$.}
  \label{table5}
\begin{tabular}{|cc|c|}
\hline
$x_{1}(0)$&$x_{2}(0)$&$x_{1}(1)$\\
\hline
$0$&$0$&$0$\\
$0$&$1$&$0$\\
$1$&$0$&$0$\\
$1$&$1$&$1$\\
\hline
\end{tabular}
\end{table}

Assuming that all four inputs are given with probability $\frac{1}{4}$, the entropy of $x_{1}(1)$ is
\begin{eqnarray*}
-\left(\frac{3}{4}\log_{2}\frac{3}{4}+\frac{1}{4}\log_{2}\frac{1}{4}\right) & \approx & 0.812.
\end{eqnarray*}
It is seen from Theorem 12 of \cite{guo22} that in total, the information
quantify of $\xvec(1)$ is at most $0.812n$ bits, which might be obtained by $[\yvec(0),\ldots,\yvec(N)]$. However, in order to observe any initial state $\xvec(0)$, we should have the information quantity of $n$ bits. Therefore, at least $n-0.812n=0.188n$ bits must be provided from $\yvec(0)$, which implies $m\geq0.188n$.
\end{proof}

Note that Theorem \ref{theorem1} gives a lower bound for the best case. In the worst case, we need $m=n$ observation nodes.

\begin{proposition}\label{proposition1}
In the worst case, we need $m=n$ observation nodes for 2-AND-OR-BN so that this BN is observable.
\end{proposition}
\begin{proof}
Consider the following BN:
\begin{eqnarray*}
x_{1}(t+1) & = & x_{1}(t)\wedge x_{n}(t),\\
x_{2}(t+1) & = & x_{2}(t)\wedge x_{n}(t),\\
&\cdots&\\
x_{n-1}(t+1) & = & x_{n-1}(t)\wedge x_{n}(t),\\
x_{n}(t+1) & = & x_{n}(t)\wedge x_{1}(t).
\end{eqnarray*}

First, note that if $x_{n}(0)=0$, $x_{i}(1)=0$ holds for all $i\in[1,n]$. Thus, we need to know $x_{i}(0)$ for all $i\in[1,n-1]$. Next, note that if $x_{1}(0)=0$, $x_{n}(1)=0$ holds. Thus, we need to know $x_{n}(0)$. Accordingly, we need to know $x_{i}(0)$ for all nodes and thus $m=n$ must holds.
\end{proof}

On the other hand, we show that there exists a 2-AND-OR-BN with $m=\frac{1}{3}n$.

\begin{proposition}\label{proposition2}
There exists a 2-AND-OR-BN that is observable with $m=\frac{1}{3}n$, where $(n \mod 3)=0$.
\end{proposition}
\begin{proof}
First, we consider the case of $n=3$ and construct the following BN:
\begin{eqnarray*}
x_{1}(t+1) & = & \overline{x_{2}(t)}\wedge x_{3}(t),\\
x_{2}(t+1) & = & x_{2}(t)\wedge x_{3}(t),\\
x_{3}(t+1) & = & \overline{x_{2}(t)}\wedge\overline{x_{3}(t)},\\
y(t) & = & x_{1}(t).
\end{eqnarray*}
Then, we have the following state transition table (Table \ref{table6}).
\begin{table}[ht!]
  \centering
  \caption{State transition table in Proposition \ref{proposition2}.}
  \label{table6}
\begin{tabular}{|ccc|ccc|ccc|ccc|}
\hline
$x_{1}(0)$&$x_{2}(0)$&$x_{3}(0)$&$x_{1}(1)$&$x_{2}(1)$&$x_{3}(1)$&$x_{1}(2)$&$x_{2}(2)$&$x_{3}(2)$&$x_{1}(3)$&$x_{2}(3)$&$x_{3}(3)$\\           \hline
$0$&$0$&$0$&$0$&$0$&$1$&$1$&$0$&$0$&$0$&$0$&$1$\\
$0$&$0$&$1$&$1$&$0$&$0$&$0$&$0$&$1$&$1$&$0$&$0$\\
$0$&$1$&$0$&$0$&$0$&$0$&$0$&$0$&$1$&$1$&$0$&$0$\\
$0$&$1$&$1$&$0$&$1$&$0$&$0$&$0$&$0$&$0$&$0$&$1$\\
$1$&$0$&$0$&$0$&$0$&$1$&$1$&$0$&$0$&$0$&$0$&$1$\\
$1$&$0$&$1$&$1$&$0$&$0$&$0$&$0$&$1$&$1$&$0$&$0$\\
$1$&$1$&$0$&$0$&$0$&$0$&$0$&$0$&$1$&$1$&$0$&$0$\\
$1$&$1$&$1$&$0$&$1$&$0$&$0$&$0$&$0$&$0$&$0$&$1$\\
\hline
\end{tabular}
\end{table}

Clearly, $\xvec(0)$ can be recovered from $y(0),y(1),y(2),y(3)$.

Next, we consider the case of $n>3$. In this case, it is enough to make $\frac{n}{3}$ copies of the above BN. Then, the number of observation nodes is $\frac{n}{3}\times1=\frac{1}{3}n$.
\end{proof}

Based on Theorem \ref{theorem1} and Proposition \ref{proposition2}, it is natural to further study whether there exists a 2-AND-OR-BN with $n>3$, which is observable with $m=1$, and then we have the following result.

\begin{proposition}\label{proposition3}
There does not exist any 2-AND-OR-BN with $n>3$ that is observable with $m=1$.
\end{proposition}
\begin{proof}
First from Example \ref{example1}, we find that $\max_{\avec^{j}}COUNT(\avec^{j})=3$. Specifically, there are three different initial states $\xvec^{1}(0)=[0,0,0],\xvec^{2}(0)=[0,1,0],\xvec^{3}(0)=[1,0,0]$, where $\xvec(1;\xvec^{1}(0))=\xvec(1;\xvec^{2}(0))=\xvec(1;\xvec^{3}(0))=[0,0,0]$. Without loss of generality, we assume $y(t)=x_{1}(t)$, and then $y(0;\xvec^{1}(0))=0,y(0;\xvec^{2}(0))=0,y(0;\xvec^{3}(0))=1,~y(t;\xvec^{1}(0))=y(t;\xvec^{2}(0))=y(t;\xvec^{3}(0))=0,t\geq1$. Clearly, we cannot distinguish between $\xvec^{1}(0)=[0,0,0]$ and $\xvec^{2}(0)=[0,1,0]$ from $y(t)=x_{1}(t)$. Now we have the following fact.

\textbf{Claim 1}: If a 2-AND-OR-BN with $m=1$ observation node is observable, then the number of identical states $\xvec(1;\xvec(0))$ for all $2^{n}$ initial states $\xvec(0)$ does not exceed 2, i.e.,
\begin{eqnarray*}
\max_{\avec^{j}}COUNT(\avec^{j})\leq2.
\end{eqnarray*}
\textbf{Proof of Claim 1}: Suppose there are three different initial states $\xvec^{1}(0),\xvec^{2}(0),\xvec^{3}(0)$ such that
\begin{eqnarray*}
\xvec(1;\xvec^{1}(0))=\xvec(1;\xvec^{2}(0))=\xvec(1;\xvec^{3}(0)).
\end{eqnarray*}
We can only discriminate these three initial states from $\yvec(0)$ because $\yvec(t;\xvec^{1}(0))=\yvec(t;\xvec^{2}(0))=\yvec(t;\xvec^{3}(0)),~t\geq1$ (Here, $\yvec(t)=[y_{1}(t),\ldots,y_{m}(t)]$, where $y_{j}(t)= x_{\phi(j)}(t),~\forall j\in[1,m]$ and $\phi(j)$ is a function from $[1,m]$ to $[1,n]$). Then if $\yvec(0)=y_{1}(0)$, i.e., $m=1$, it is clear that $\xvec^{1}(0),\xvec^{2}(0),\xvec^{3}(0)$ cannot be determined from $y_{1}(0)$, because the value of $y_{1}(0)$ can only be 0 or 1.

Then we find that $\max_{\avec^{j}}COUNT(\avec^{j})>2$ for the 2-AND-OR-BN with $n=4$ in Example \ref{example2} and we have the following fact.

\textbf{Claim 2}: For any 2-AND-OR-BN with $n>3$, the number of identical states $\xvec(1;\xvec(0))$ for all $2^{n}$ initial states $\xvec(0)$ must exceed 2, i.e.,
\begin{eqnarray*}
\min\max_{\avec^{j}}COUNT(\avec^{j})>2.
\end{eqnarray*}
\textbf{Proof of Claim 2}: We only consider the case where $f_{i}$ is AND of two literals, where the number of ones (resp. zeros) corresponding to $x_{i}(1)$ for all $2^{n}$ initial states $\xvec(0)$ is $\frac{1}{4}\times2^{n}$ (resp. $\frac{3}{4}\times2^{n}$). Notably, the same result holds for any AND or OR functions of two literals.

First, we still consider the case of $n=4$, where it is clear that the number of ones (resp. zeros) corresponding to $x_{i}(1),~i\in[1,4]$ for all $16$ initial states $\xvec(0)$ is $4$ (resp. $12$) and we assume that there exists a 2-AND-OR-BN with $n=4$ that $\max_{\avec^{j}}COUNT(\avec^{j})=2$. If the number of states $\xvec(1;\xvec(0),x_{3}(1)=0,x_{4}(1)=0)$ exceeds 8, then $\max_{\avec^{j}}COUNT(\avec^{j})>2$. As shown in the table below (Table \ref{table7}), assuming that the number of states $\xvec(1;\xvec(0),x_{3}(1)=0,x_{4}(1)=0)$ is 9, then in this case, $[x_{1}(1),x_{2}(1)]$ has only $4$ different values, i.e., $[0,0],[0,1],[1,0],[1,1]$, so $\max_{\avec^{j}}COUNT(\avec^{j})>2$.
\begin{table}[ht!]
  \centering
  \caption{When the number of states $\xvec(1;\xvec(0),x_{3}(1)=0,x_{4}(1)=0)$ is 9, all corresponding states $\xvec(1;\xvec(0),x_{3}(1)=0,x_{4}(1)=0)$.}
  \label{table7}
\begin{tabular}{|cccc|}
\hline
$x_{1}(1)$&$x_{2}(1)$&$x_{3}(1)$&$x_{4}(1)$\\
\hline
$0$&$0$&$0$&$0$\\
$0$&$0$&$0$&$0$\\
$0$&$1$&$0$&$0$\\
$0$&$1$&$0$&$0$\\
$1$&$0$&$0$&$0$\\
$1$&$0$&$0$&$0$\\
$1$&$1$&$0$&$0$\\
$1$&$1$&$0$&$0$\\
$0$&$0$&$0$&$0$\\
\hline
\end{tabular}
\end{table}

Suppose that $\xvec(1;\xvec(0),x_{3}(1)=0,x_{4}(1)=0)$ is $8$.
In order to ensure that $\max_{\avec^{j}}COUNT(\avec^{j})=2$ in these $8$ states, we have the following table (Table \ref{table8}).
\begin{table}[ht!]
  \centering
  \caption{When the number of states $\xvec(1;\xvec(0),x_{3}(1)=0,x_{4}(1)=0)$ is 8, all corresponding states $\xvec(1;\xvec(0),x_{3}(1)=0,x_{4}(1)=0)$.}
  \label{table8}
\begin{tabular}{|cccc|}
\hline
$x_{1}(1)$&$x_{2}(1)$&$x_{3}(1)$&$x_{4}(1)$\\
\hline
$0$&$0$&$0$&$0$\\
$0$&$0$&$0$&$0$\\
$0$&$1$&$0$&$0$\\
$0$&$1$&$0$&$0$\\
$1$&$0$&$0$&$0$\\
$1$&$0$&$0$&$0$\\
$1$&$1$&$0$&$0$\\
$1$&$1$&$0$&$0$\\
\hline
\end{tabular}
\end{table}
However, for 16 $\xvec(1;\xvec(0))$,
since the number of ones of $x_{i}(1),~i\in[1,4]$ is 4,
as shown in the table below (Table \ref{table9}),
$COUNT(\avec^{j})=4$ must hold for
both $\avec^{j}=[0,0,0,1]$ and
$\avec^{j}=[0,0,1,0]$.

\begin{table}[ht!]
  \centering
  \caption{All 16 states $\xvec(1;\xvec(0))$.}
  \label{table9}
\begin{tabular}{|cccc|}
\hline
$x_{1}(1)$&$x_{2}(1)$&$x_{3}(1)$&$x_{4}(1)$\\
\hline
$0$&$0$&$0$&$1$\\
$0$&$0$&$0$&$1$\\
$0$&$0$&$0$&$1$\\
$0$&$0$&$0$&$1$\\
$0$&$0$&$1$&$0$\\
$0$&$0$&$1$&$0$\\
$0$&$0$&$1$&$0$\\
$0$&$0$&$1$&$0$\\
$0$&$0$&$0$&$0$\\
$0$&$0$&$0$&$0$\\
$0$&$1$&$0$&$0$\\
$0$&$1$&$0$&$0$\\
$1$&$0$&$0$&$0$\\
$1$&$0$&$0$&$0$\\
$1$&$1$&$0$&$0$\\
$1$&$1$&$0$&$0$\\
\hline
\end{tabular}
\end{table}
Furthermore, we can find that the number of states $\xvec(1;\xvec(0),x_{3}(1)=0,x_{4}(1)=0)$ is at least $8$ because, for 16 $\xvec(1;\xvec(0))$,
the number of ones of $x_{3}(1)$ (resp. $x_{4}(1)$) is 4.
Hence, we find that $\max_{\avec^{j}}COUNT(\avec^{j})>2$ for any 2-AND-OR-BN with $n=4$.

We generalize the above proof for any $n > 3$.
Suppose there exists a 2-AND-OR-BN with $n>3$ that $\max_{\avec^{j}}COUNT(\avec^{j})\leq2$.

Among $2^{n}$ states $\xvec(1;\xvec(0))$, there must be $\frac{1}{4}\times2^{n}$ states $\xvec(1;\xvec(0),x_{n}(1)=1)$, $\frac{1}{4}\times2^{n}$ states $\xvec(1;\xvec(0),x_{n-1}(1)=1,x_{n}(1)=0)$ and $\frac{2}{4}\times2^{n}$ states $\xvec(1;\xvec(0),x_{n-1}(1)=0,x_{n}(1)=0)$ by the following reason.
\begin{itemize}
\item For $[x_{1}(1),\ldots,x_{n-2}(1)]$, there are at most $2^{n-2}$ different vectors.
Then, in order to ensure that $\max_{\avec^{j}}COUNT(\avec^{j})\leq2$,
the number of states $\xvec(1;\xvec(0),x_{n-1}(1)=0,x_{n}(1)=0)$ does not exceed $2^{n-1}$ (In the case of $n=4$, it means that the number of $\xvec(1;\xvec(0),x_{3}(1)=0,x_{4}(1)=0)$ does not exceed $8$).
\item For $2^n$ vectors $\xvec(0)$, the number of ones of $x_{n-1}(1)$
(resp. $x_{n}(1)$) in $\xvec(1;\xvec(0))$ is $2^{n-2}$.
Thus,
the number of states $\xvec(1;\xvec(0),x_{n-1}(1)=0,x_{n}(1)=0)$ is at least
$2^{n-1}$.
\end{itemize}
Then, since
the number of states $\xvec(1;\xvec(0),x_{n-1}(1)=0,x_{n}(1)=0)$ is $2^{n-1}$
and $\max_{\avec^{j}}COUNT(\avec^{j})\leq2$,
the number of states $\xvec(1;\xvec(0),\xvec(1)=\avec)$ must be two
for each of $2^{n-2}$ $\avec$'s with $a_{n-1}=0$ and $a_n=0$.
This means that we cannot use 1 for any $x_{i}$ with $i < n-1$
in any $\xvec(1)$ with $x_{n-1}(1)=1$ and $x_{n}(1)=0$.
Therefore, we have
$COUNT(\avec)=2^{n-2}$ for $\avec=[0,0,\ldots,0,1,0]$,
which contradicts the assumption.

Finally, based on Claim 1 and Claim 2, there does not exist any 2-AND-OR-BN with $n>3$ that is observable with $m=1$.
\end{proof}

\begin{remark}
Compared to Theorem \ref{theorem1}, it seems that Proposition \ref{proposition3} is useful only when $0.188n<1$ (i.e., $n<6$). However, the technique mentioned in Proposition \ref{proposition3} for counting identical states $\xvec(1;\xvec(0))$ (i.e., finding $\max_{\avec^{i}}COUNT(\avec^{i})$) might be useful to show other lower bounds of the number of observation nodes $m$.
\end{remark}

Now we generalize Claim 1 and present a new technique to study the lower bounds of the number of observation nodes $m$, which is applicable to any BNs, as shown below.

\begin{proposition}\label{proposition4}
For any BN,
\begin{eqnarray*}
2^{m}\geq\max_{\avec^{i}}COUNT(\avec^{i})
\end{eqnarray*}
must hold so that the BN is observable.
\end{proposition}
\begin{proof}
Suppose $\max_{\avec^{j}}COUNT(\avec^{j})=l$, i.e., there are $l$ different initial states $\xvec^{1}(0),\ldots,\xvec^{l}(0)$ such that $\xvec(1;\xvec^{1}(0))=\cdots=\xvec(1;\xvec^{l}(0))$. We can only discriminate these $l$ different initial states $\xvec^{1}(0),\ldots,\xvec^{l}(0)$ from $\yvec(0)$ because $\yvec(t;\xvec^{1}(0))=\cdots=\yvec(t;\xvec^{l}(0)),~t\geq1$. Then if $\xvec^{1}(0),\ldots,\xvec^{l}(0)$ can be distinguished from $\yvec(0)=[y_{1}(0),\ldots,y_{m}(0)]$, there are $l$ different values $\yvec(0;\xvec^{1}(0)),\ldots,\yvec(0;\xvec^{l}(0))$. However, the number of different values of $\yvec(0)$ is at most $2^{m}$ and $2^{m}<l$, which implies that $\xvec^{1}(0),\ldots,\xvec^{l}(0)$ cannot be determined from $\yvec(0)=[y_{1}(0),\ldots,y_{m}(0)]$.
\end{proof}

Then we focus on the case where each $f_{i},~i\in[1,n]$ is AND or OR of $K>2$ literals. In other words, for each $f_{i},~i\in[1,n]$, we only consider the following two forms, $\ell_{i_{1}}(t)\wedge \ell_{i_{2}}(t)\wedge\cdots\wedge \ell_{i_{K}}(t)$ or $\ell_{i_{1}}(t)\vee \ell_{i_{2}}(t)\vee\cdots\vee \ell_{i_{K}}(t)$, where $\ell_{i}\in\{x_{1},\overline{x_{1}},x_{2},\overline{x_{2}},\ldots,x_{n},\overline{x_{n}}\}$.
Similarly, we call these BNs $K$-AND-OR-BN.

We extend the above results to $K$-AND-OR-BN with arbitrary fixed $K>2$ as follows.

\begin{theorem}\label{theorem2}
For any K-AND-OR-BN, $m\geq \left[(1-K)+\frac{2^{K}-1}{2^{K}}\log_{2}(2^{K}-1)\right]n$ must hold so that the BN is observable.
\end{theorem}
\begin{proof}
Similarly, consider the case of $x_{1}(t+1)=x_{1}(t)\wedge x_{2}(t)\wedge\cdots\wedge x_{K}(t)$, where the same result holds for any AND or OR functions of $K$ literals. Then, the state transitions from $\xvec(0)$ to $x_{1}(1)$ are given as below (Table \ref{table10}).
\begin{table}[ht!]
  \centering
  \caption{State transition table under the case of $x_{1}(t+1)=x_{1}(t)\wedge x_{2}(t)\wedge\cdots\wedge x_{K}(t)$.}
  \label{table10}
\begin{tabular}{|ccccc|c|}
\hline
$x_{1}(0)$&$x_{2}(0)$&$\cdots$&$x_{K-1}(0)$&$x_{K}(0)$&$x_{1}(1)$\\
\hline
$0$&$0$&$\cdots$&$0$&$0$&$0$\\
$0$&$0$&$\cdots$&$0$&$1$&$0$\\
$\vdots$&$\vdots$&$\ddots$&$\vdots$&$\vdots$&$\vdots$\\
$1$&$1$&$\cdots$&$1$&$0$&$0$\\
$1$&$1$&$\cdots$&$1$&$1$&$1$\\
\hline
\end{tabular}
\end{table}
Assuming that all $2^{K}$ inputs are given with probability $\frac{1}{2^{K}}$, the entropy of $x_{1}(1)$ is
\begin{eqnarray*}
-\left[\frac{1}{2^{K}}\log_{2}\left(\frac{1}{2^{K}}\right)+\frac{2^{K}-1}{2^{K}}\log_{2}\left(\frac{2^{K}-1}{2^{K}}\right)\right]=K-\frac{2^{K}-1}{2^{K}}\log_{2}(2^{K}-1).
\end{eqnarray*}
In total, the information quantity of $\xvec(1)$ is at most $\left[K-\frac{2^{K}-1}{2^{K}}\log_{2}(2^{K}-1)\right]n$ bits, which might be obtained by $[\yvec(1),\ldots,\yvec(N)]$. However, in order to observe any $\xvec(0)$, we should have the information quantity of $n$ bits. Therefore, at least $n-\left[K-\frac{2^{K}-1}{2^{K}}\log_{2}(2^{K}-1)\right]n$ bits must be provided from $\yvec(0)$, which implies $m\geq \left[(1-K)+\frac{2^{K}-1}{2^{K}}\log_{2}(2^{K}-1)\right]n$.
\end{proof}

In the worst case, we need $m=n$ observation nodes.

\begin{proposition}\label{proposition5}
In the worst case, we need $m=n$ observation nodes for $K$-AND-OR-BN with arbitrary fixed $K>2$ so that this BN is observable.
\end{proposition}
\begin{proof}
Consider the following BN:
\begin{eqnarray*}
x_{1}(t+1) & = & x_{1}(t)\wedge x_{2}(t)\wedge\cdots\wedge x_{K-1}(t)\wedge x_{K}(t),\\
x_{2}(t+1) & = & \overline{x_{1}(t)}\wedge x_{2}(t)\wedge\cdots\wedge x_{K-1}(t)\wedge x_{K}(t),\\
& \cdots &\\
x_{K}(t+1) & = & x_{1}(t)\wedge x_{2}(t)\wedge\cdots\wedge \overline{x_{K-1}(t)}\wedge x_{K}(t),\\
x_{K+1}(t+1) & = & x_{2}(t)\wedge x_{3}(t)\wedge\cdots\wedge x_{K}(t)\wedge x_{K+1}(t),
\end{eqnarray*}
where $K>2$. Here, the number of $\xvec(1;\xvec(0))=[0,0,\ldots,0]$ for all $2^{n}$ initial states $\xvec(0)$ is $2^{K+1}-2K$. Then the number of observation nodes is $K+1$, because if this $K$-AND-OR-BN with $m$ observation nodes is observable, then $2^{m}\geq2^{K+1}-2K,~K>2$, i.e., $m=K+1$.
\end{proof}

On the other hand, we show that there exists a $K$-AND-OR-BN with $m=\left(\frac{2^{K}-K-1}{2^{K}-1}\right)n$.

\begin{theorem}\label{theorem}
For any $n$ and $K>2$, there exists a $K$-AND-OR-BN that is observable with $m=\left(\frac{2^{K}-K-1}{2^{K}-1}\right)n$, where $(n \mod 2^{K}-1)=0$.
\end{theorem}
\begin{proof}
First, we consider the case of $n=2^{K}-1$ and construct the following BN:
\begin{eqnarray*}
x_{1}(t+1) & = & \overline{x_{2^{K}-K}(t)}\wedge\overline{x_{2^{K}-K+1}(t)}\wedge\cdots\wedge\overline{x_{2^{K}-3}(t)}\wedge\overline{x_{2^{K}-2}(t)}\wedge\overline{x_{2^{K}-1}(t)},\\
x_{2}(t+1) & = & \overline{x_{2^{K}-K}(t)}\wedge\overline{x_{2^{K}-K+1}(t)}\wedge\cdots\wedge\overline{x_{2^{K}-3}(t)}\wedge\overline{x_{2^{K}-2}(t)}\wedge x_{2^{K}-1}(t),\\
& \cdots &\\
x_{K+1}(t+1) & = & x_{2^{K}-K}(t)\wedge\overline{x_{2^{K}-K+1}(t)}\wedge\cdots\wedge\overline{x_{2^{K}-3}(t)}\wedge\overline{x_{2^{K}-2}(t)}\wedge\overline{x_{2^{K}-1}(t)},\\
x_{K+2}(t+1) & = & \overline{x_{2^{K}-K}(t)}\wedge \overline{x_{2^{K}-K+1}(t)}\wedge\cdots\wedge\overline{x_{2^{K}-3}(t)}\wedge x_{2^{K}-2}(t)\wedge x_{2^{K}-1}(t),\\
& \cdots &\\
x_{2^{K}-K-1}(t+1) & = & x_{2^{K}-K}(t)\wedge x_{2^{K}-K+1}(t)\wedge\cdots\wedge x_{2^{K}-3}(t)\wedge\overline{x_{2^{K}-2}(t)}\wedge\overline{x_{2^{K}-1}(t)},\\
x_{2^{K}-K}(t+1) & = & x_{2^{K}-K}(t)\wedge x_{2^{K}-K+1}(t)\wedge\cdots\wedge x_{2^{K}-3}(t)\wedge x_{2^{K}-2}(t)\wedge x_{2^{K}-1}(t),\\
x_{2^{K}-K+1}(t+1) & = & \overline{x_{2^{K}-K}(t)}\wedge x_{2^{K}-K+1}(t)\wedge\cdots\wedge x_{2^{K}-3}(t)\wedge x_{2^{K}-2}(t)\wedge x_{2^{K}-1}(t),\\
& \cdots &\\
x_{2^{K}-1}(t+1) & = & x_{2^{K}-K}(t)\wedge x_{2^{K}-K+1}(t)\wedge\cdots\wedge x_{2^{K}-3}(t)\wedge\overline{x_{2^{K}-2}(t)}\wedge x_{2^{K}-1}(t),\\
y_{1}(t) & = & x_{1}(t),\\
y_{2}(t) & = & x_{2}(t),\\
& \cdots &\\
y_{2^{K}-K-1}(t) & = & x_{2^{K}-K-1}(t).
\end{eqnarray*}
Then, we have the following state transition table (Table \ref{table11}).
\begin{table}[ht!]
  \centering
  \caption{State transition table in Theorem \ref{theorem}.}
  \label{table11}
\resizebox{\linewidth}{!}{
\begin{tabular}{|cccccccc|cccccccccccc|}
\hline
$\cdots$&$x_{2^{K}-K}(0)$&$x_{2^{K}-K+1}(0)$&$x_{2^{K}-K+2}(0)$&$\cdots$&$x_{2^{K}-3}(0)$&$x_{2^{K}-2}(0)$&$x_{2^{K}-1}(0)$&$x_{1}(1)$&$x_{2}(1)$&$\cdots$&$x_{2^{K}-K-2}(1)$&$x_{2^{K}-K-1}(1)$&$x_{2^{K}-K}(1)$&$x_{2^{K}-K+1}(1)$&$x_{2^{K}-K+2}(1)$&$\cdots$&$x_{2^{K}-3}(1)$&$x_{2^{K}-2}(1)$&$x_{2^{K}-1}(1)$\\           \hline
$\cdots$&$0$&$0$&$0$&$\cdots$&$0$&$0$&$0$&$1$&$0$&$\cdots$&$0$&$0$&$0$&$0$&$0$&$\cdots$&$0$&$0$&$0$\\
$\cdots$&$0$&$0$&$0$&$\cdots$&$0$&$0$&$1$&$0$&$1$&$\cdots$&$0$&$0$&$0$&$0$&$0$&$\cdots$&$0$&$0$&$0$\\
$\cdots$&$\vdots$&$\vdots$&$\vdots$&$\ddots$&$\vdots$&$\vdots$&$\vdots$&$\vdots$&$\vdots$&$\ddots$&$\vdots$&$\vdots$&$\vdots$&$\vdots$&$\vdots$&$\ddots$&$\vdots$&$\vdots$&$\vdots$\\
$\cdots$&$1$&$1$&$1$&$\cdots$&$1$&$0$&$0$&$0$&$0$&$\cdots$&$0$&$1$&$0$&$0$&$0$&$\cdots$&$0$&$0$&$0$\\
$\cdots$&$1$&$1$&$1$&$\cdots$&$1$&$1$&$1$&$0$&$0$&$\cdots$&$0$&$0$&$1$&$0$&$0$&$\cdots$&$0$&$0$&$0$\\
$\cdots$&$0$&$1$&$1$&$\cdots$&$1$&$1$&$1$&$0$&$0$&$\cdots$&$0$&$0$&$0$&$1$&$0$&$\cdots$&$0$&$0$&$0$\\
$\cdots$&$1$&$0$&$1$&$\cdots$&$1$&$1$&$1$&$0$&$0$&$\cdots$&$0$&$0$&$0$&$0$&$1$&$\cdots$&$0$&$0$&$0$\\
$\ddots$&$\vdots$&$\vdots$&$\vdots$&$\ddots$&$\vdots$&$\vdots$&$\vdots$&$\vdots$&$\vdots$&$\ddots$&$\vdots$&$\vdots$&$\vdots$&$\vdots$&$\vdots$&$\ddots$&$\vdots$&$\vdots$&$\vdots$\\
$\cdots$&$1$&$1$&$1$&$\cdots$&$0$&$1$&$1$&$0$&$0$&$\cdots$&$0$&$0$&$0$&$0$&$0$&$\cdots$&$0$&$1$&$0$\\
$\cdots$&$1$&$1$&$1$&$\cdots$&$1$&$0$&$1$&$0$&$0$&$\cdots$&$0$&$0$&$0$&$0$&$0$&$\cdots$&$0$&$0$&$1$\\
$\cdots$&$1$&$1$&$1$&$\cdots$&$1$&$1$&$0$&$0$&$0$&$\cdots$&$0$&$0$&$0$&$0$&$0$&$\cdots$&$0$&$0$&$0$\\
\hline
\end{tabular}}
\end{table}

Clearly, $\xvec(0)$ can be recovered from $\yvec(0),\yvec(1),\yvec(2)$.

Next, we consider the case of $n>2^{K}-1$. In this case, it is enough to make $\frac{n}{2^{K}-1}$ copies of the above BN. Then, the number of observation nodes is $\frac{n}{2^{K}-1}\times(2^{K}-K-1)=\left(\frac{2^{K}-K-1}{2^{K}-1}\right)n$.
\end{proof}

\begin{example}\label{exmp}
Consider the following BN:
\begin{eqnarray*}
x_{1}(t+1) & = & \overline{x_{12}(t)}\wedge \overline{x_{13}(t)}\wedge \overline{x_{14}(t)}\wedge \overline{x_{15}(t)},\\
x_{2}(t+1) & = & \overline{x_{12}(t)}\wedge \overline{x_{13}(t)}\wedge \overline{x_{14}(t)}\wedge x_{15}(t),\\
x_{3}(t+1) & = & \overline{x_{12}(t)}\wedge \overline{x_{13}(t)}\wedge x_{14}(t)\wedge \overline{x_{15}(t)},\\
x_{4}(t+1) & = & \overline{x_{12}(t)}\wedge x_{13}(t)\wedge \overline{x_{14}(t)}\wedge \overline{x_{15}(t)},\\
x_{5}(t+1) & = & x_{12}(t)\wedge \overline{x_{13}(t)}\wedge \overline{x_{14}(t)}\wedge \overline{x_{15}(t)},\\
x_{6}(t+1) & = & \overline{x_{12}(t)}\wedge \overline{x_{13}(t)}\wedge x_{14}(t)\wedge x_{15}(t),\\
x_{7}(t+1) & = & \overline{x_{12}(t)}\wedge x_{13}(t)\wedge \overline{x_{14}(t)}\wedge x_{15}(t),\\
x_{8}(t+1) & = & x_{12}(t)\wedge \overline{x_{13}(t)}\wedge \overline{x_{14}(t)}\wedge x_{15}(t),\\
x_{9}(t+1) & = & \overline{x_{12}(t)}\wedge x_{13}(t)\wedge x_{14}(t)\wedge \overline{x_{15}(t)},\\
x_{10}(t+1) & = & x_{12}(t)\wedge \overline{x_{13}(t)}\wedge x_{14}(t)\wedge \overline{x_{15}(t)},\\
x_{11}(t+1) & = & x_{12}(t)\wedge x_{13}(t)\wedge \overline{x_{14}(t)}\wedge \overline{x_{15}(t)},\\
x_{12}(t+1) & = & x_{12}(t)\wedge x_{13}(t)\wedge x_{14}(t)\wedge x_{15}(t),\\
x_{13}(t+1) & = & \overline{x_{12}(t)}\wedge x_{13}(t)\wedge x_{14}(t)\wedge x_{15}(t),\\
x_{14}(t+1) & = & x_{12}(t)\wedge \overline{x_{13}(t)}\wedge x_{14}(t)\wedge x_{15}(t),\\
x_{15}(t+1) & = & x_{12}(t)\wedge x_{13}(t)\wedge \overline{x_{14}(t)}\wedge x_{15}(t),\\
y_{1}(t) & = & x_{1}(t),\\
y_{2}(t) & = & x_{2}(t),\\
& \cdots &\\
y_{11}(t) & = & x_{11}(t).
\end{eqnarray*}
Then, we have the following state transition table (Table \ref{table12}).
\begin{table}[ht!]
  \centering
  \caption{State transition table in Example \ref{exmp}.}
  \label{table12}
\resizebox{\linewidth}{!}{
\begin{tabular}{|ccccc|ccccccccccccccc|}
\hline
$\cdots$&$x_{12}(0)$&$x_{13}(0)$&$x_{14}(0)$&$x_{15}(0)$&$x_{1}(1)$&$x_{2}(1)$&$x_{3}(1)$&$x_{4}(1)$&$x_{5}(1)$&$x_{6}(1)$&$x_{7}(1)$&$x_{8}(1)$&$x_{9}(1)$&$x_{10}(1)$&$x_{11}(1)$&$x_{12}(1)$&$x_{13}(1)$&$x_{14}(1)$&$x_{15}(1)$\\           \hline
$\cdots$&$0$&$0$&$0$&$0$&$1$&$0$&$0$&$0$&$0$&$0$&$0$&$0$&$0$&$0$&$0$&$0$&$0$&$0$&$0$\\
$\cdots$&$0$&$0$&$0$&$1$&$0$&$1$&$0$&$0$&$0$&$0$&$0$&$0$&$0$&$0$&$0$&$0$&$0$&$0$&$0$\\
$\cdots$&$0$&$0$&$1$&$0$&$0$&$0$&$1$&$0$&$0$&$0$&$0$&$0$&$0$&$0$&$0$&$0$&$0$&$0$&$0$\\
$\cdots$&$0$&$1$&$0$&$0$&$0$&$0$&$0$&$1$&$0$&$0$&$0$&$0$&$0$&$0$&$0$&$0$&$0$&$0$&$0$\\
$\cdots$&$1$&$0$&$0$&$0$&$0$&$0$&$0$&$0$&$1$&$0$&$0$&$0$&$0$&$0$&$0$&$0$&$0$&$0$&$0$\\
$\cdots$&$0$&$0$&$1$&$1$&$0$&$0$&$0$&$0$&$0$&$1$&$0$&$0$&$0$&$0$&$0$&$0$&$0$&$0$&$0$\\
$\cdots$&$0$&$1$&$0$&$1$&$0$&$0$&$0$&$0$&$0$&$0$&$1$&$0$&$0$&$0$&$0$&$0$&$0$&$0$&$0$\\
$\cdots$&$1$&$0$&$0$&$1$&$0$&$0$&$0$&$0$&$0$&$0$&$0$&$1$&$0$&$0$&$0$&$0$&$0$&$0$&$0$\\
$\cdots$&$0$&$1$&$1$&$0$&$0$&$0$&$0$&$0$&$0$&$0$&$0$&$0$&$1$&$0$&$0$&$0$&$0$&$0$&$0$\\
$\cdots$&$1$&$0$&$1$&$0$&$0$&$0$&$0$&$0$&$0$&$0$&$0$&$0$&$0$&$1$&$0$&$0$&$0$&$0$&$0$\\
$\cdots$&$1$&$1$&$0$&$0$&$0$&$0$&$0$&$0$&$0$&$0$&$0$&$0$&$0$&$0$&$1$&$0$&$0$&$0$&$0$\\
$\cdots$&$1$&$1$&$1$&$1$&$0$&$0$&$0$&$0$&$0$&$0$&$0$&$0$&$0$&$0$&$0$&$1$&$0$&$0$&$0$\\
$\cdots$&$0$&$1$&$1$&$1$&$0$&$0$&$0$&$0$&$0$&$0$&$0$&$0$&$0$&$0$&$0$&$0$&$1$&$0$&$0$\\
$\cdots$&$1$&$0$&$1$&$1$&$0$&$0$&$0$&$0$&$0$&$0$&$0$&$0$&$0$&$0$&$0$&$0$&$0$&$1$&$0$\\
$\cdots$&$1$&$1$&$0$&$1$&$0$&$0$&$0$&$0$&$0$&$0$&$0$&$0$&$0$&$0$&$0$&$0$&$0$&$0$&$1$\\
$\cdots$&$1$&$1$&$1$&$0$&$0$&$0$&$0$&$0$&$0$&$0$&$0$&$0$&$0$&$0$&$0$&$0$&$0$&$0$&$0$\\
\hline
\end{tabular}}
\end{table}

Clearly, $\xvec(0)$ can be recovered from $\yvec(0),\yvec(1),\yvec(2)$, where $\yvec(t)=[y_{1}(t),\ldots,y_{11}(t)]$.
\end{example}

\begin{remark}
Comparing the coefficient $(1-K)+\frac{2^{K}-1}{2^{K}}\log_{2}(2^{K}-1)$ in Theorem \ref{theorem2} with the coefficient $\frac{2^{K}-K-1}{2^{K}-1}$ in Theorem \ref{theorem}, we notice that $\left(\frac{2^{K}-K-1}{2^{K}-1}\right)n>\left[(1-K)+\frac{2^{K}-1}{2^{K}}\log_{2}(2^{K}-1)\right]n$ always holds regardless of $n$, and we illustrate it as follows (Fig~\ref{1}).

\begin{figure}[htb]
\centering
\includegraphics[width=0.6\textwidth]{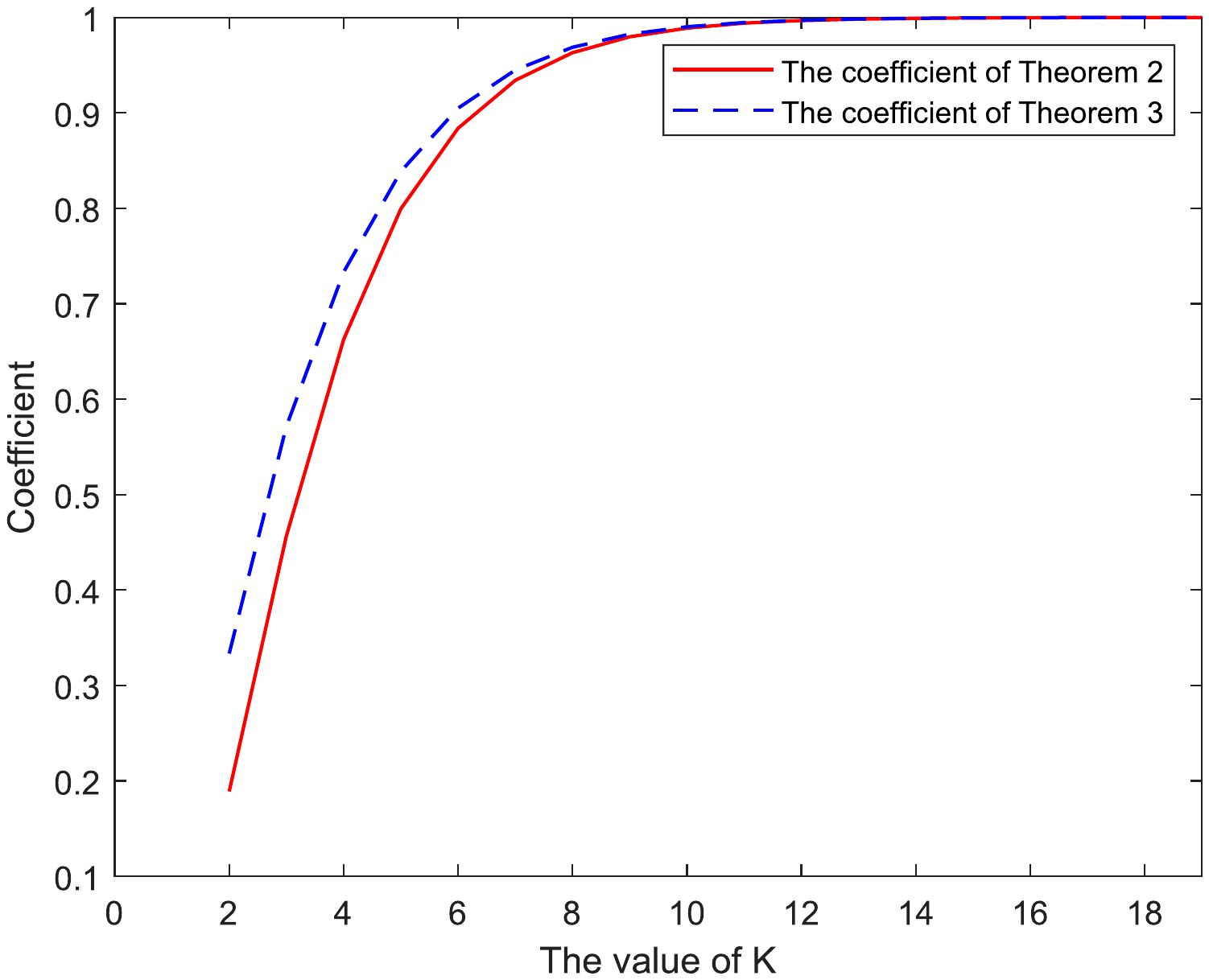}
\caption{Comparison between the coefficient $(1-K)+\frac{2^{K}-1}{2^{K}}\log_{2}(2^{K}-1)$ in Theorem \ref{theorem2} and the coefficient $\frac{2^{K}-K-1}{2^{K}-1}$ in Theorem \ref{theorem}. Here, we find that $\frac{2^{K}-K-1}{2^{K}-1}>(1-K)+\frac{2^{K}-1}{2^{K}}\log_{2}(2^{K}-1)$.}
\label{1}
\end{figure}

Let $I_{K}=(2^{K}-1)\left[K-\frac{2^{K}-1}{2^{K}}\log_{2}(2^K-1)\right]$. Then we can find that $I_{K}>K$ for any $K > 1$, as shown below.
\begin{eqnarray*}
I_K & = & (2^K-1) \left[K-\frac{2^{K}-1}{2^{K}}
\left(\log_2(1-{\frac{1}{2^K}}) + \log_2 (2^K)\right) \right]\\
& = & (2^K-1)\left[K - \frac{2^{K}-1}{2^{K}}\log_2(1-{\frac{1}{2^K}}) -
\frac{2^{K}-1}{2^{K}}\cdot K \right]\\
& = & (2^K-1)\left[{\frac{K}{2^K}} -
\frac{2^{K}-1}{2^{K}}\log_2(1-{\frac{1}{2^K}}) \right]\\
& \geq & (2^K-1)\left[{\frac{K}{2^K}} +
\frac{2^{K}-1}{2^{K}} \cdot {\frac{1}{\ln 2}} \cdot {\frac {1}{2^K}}
\right]\\
& = & K(1-{\frac{1}{2^K}}) + (2^K-1)\cdot\frac{2^{K}-1}{2^{K}} \cdot {\frac{1}{\ln2}}\cdot {\frac{1}{2^K}}\\
& > & K - {\frac{K}{2^K}} + 1.442(1-{\frac{1}{2^K}})^2 \\
& > & K,
\end{eqnarray*}
where $-\ln(1-x) \geq x,~0<x<1$. Hence, one has
\begin{eqnarray*}
K & < & (2^{K}-1)\left[K-\frac{2^{K}-1}{2^{K}}\log_{2}(2^K-1)\right]\\
1-\frac{K}{2^{K}-1} & > & 1-\left[K-\frac{2^{K}-1}{2^{K}}\log_{2}(2^K-1)\right],\\
\end{eqnarray*}
which means that $\left(\frac{2^{K}-K-1}{2^{K}-1}\right)n>\left[(1-K)+\frac{2^{K}-1}{2^{K}}\log_{2}(2^{K}-1)\right]n$ always holds regardless of $n$.
\end{remark}
\section{Boolean networks consisting of XOR functions}
In this section, we study the upper and lower bounds of the number of observation nodes $m$ in BNs consisting of XOR functions.

We first consider BNs in which an XOR function with two inputs is assigned to each node, which is referred to as a 2-XOR-BN.
\begin{proposition}\label{proposition7}
There exists a 2-XOR-BN that is observable with $m=1$.
\end{proposition}
\begin{proof}
Consider the following BN:
\begin{eqnarray*}
x_{i}(t+1) & = & x_{i}(t)\oplus x_{i+1}(t),\quad \forall i\in[1,n-1],\\
x_{n}(t+1) & = & x_{n}(t)\oplus x_{1}(t),\\
y(t) & = & x_{1}(t).
\end{eqnarray*}

We first prove that if $\yvec(t;\xvec^{1}(0))=\yvec(t;\xvec^{2}(0)),~t\in[0,n-1]$, then $\xvec^{1}(0)=\xvec^{2}(0)$. In other words, if  $[y^{1}(0),y^{1}(1),\ldots,y^{1}(n-1)]=[y^{2}(0),y^{2}(1),\ldots,y^{2}(n-1)]$, then $[x_{1}^{1}(0),x_{2}^{1}(0),\ldots,x_{n}^{1}(0)]=[x_{1}^{2}(0),x_{2}^{2}(0),\ldots,x_{n}^{2}(0)]$.

Clearly, if $\yvec(0;\xvec^{1}(0))=\yvec(0;\xvec^{2}(0))$, we have $x_{1}^{1}(0)=x_{1}^{2}(0)$. Suppose that if $[y^{1}(0),y^{1}(1),\ldots,y^{1}(j-1)]=[y^{2}(0),y^{2}(1),\ldots,y^{2}(j-1)]$, we have $[x_{1}^{1}(0),x_{2}^{1}(0),\ldots,x_{j}^{1}(0)]=[x_{1}^{2}(0),x_{2}^{2}(0),\ldots,x_{j}^{2}(0)]$.

When $[y^{1}(0),y^{1}(1),\ldots,y^{1}(j)]=[y^{2}(0),y^{2}(1),\ldots,y^{2}(j)]$, because
\begin{eqnarray*}
y(1) & = & x_{1}(0)\oplus x_{2}(0)\\
y(2) & = & x_{1}(0)\oplus x_{3}(0)\\
y(3) & = & x_{1}(0)\oplus x_{2}(0)\oplus x_{3}(0)\oplus x_{4}(0)\\
y(4) & = & x_{1}(0)\oplus x_{5}(0)\\
y(5) & = & x_{1}(0)\oplus x_{2}(0)\oplus x_{5}(0)\oplus x_{6}(0)\\
y(6) & = & x_{1}(0)\oplus x_{3}(0)\oplus x_{5}(0)\oplus x_{7}(0)\\
 & \cdots & \\
y(j) & = & x_{1}(0)\oplus\cdots\oplus x_{j+1}(0)
\end{eqnarray*}
and $[x_{1}^{1}(0),x_{2}^{1}(0),\ldots,x_{j}^{1}(0)]=[x_{1}^{2}(0),x_{2}^{2}(0),\ldots,x_{j}^{2}(0)]$,
we have $x_{j+1}^{1}(0)=x_{j+1}^{2}(0)$.

Hence, if $\yvec(t;\xvec^{1}(0))=\yvec(t;\xvec^{2}(0)),~t\in[0,n-1]$, then $\xvec^{1}(0)=\xvec^{2}(0)$, which implies that any $\xvec(0)$ can be determined from $y(0),y(1),\ldots,y(n-1)$, where $y(t)=x_{1}(t)$.
\end{proof}

\begin{remark}
The above result does not contradict Proposition 20 in \cite{guo22} because we get 1 bit information at time $t=0$.
\end{remark}

\begin{remark}
Of course, there exists a 2-XOR-BN that is observable with $m>1$. For example, consider the BN specified by $x_{i}(t+1)=x_{i}(t)\oplus x_{i+1}(t)$ for odd $i$, $x_{i}(t+1)=x_{i-1}(t)\oplus x_{i}(t)$ for even $i$. Then, we need $\frac{1}{2}n$ observation nodes.
\end{remark}

\begin{remark}\label{remark4}
In the worst case, we need $m=n$ observation nodes for 2-XOR-BN so that this BN is observable. For example, consider the BN specified by $x_{i}(t+1)=x_{1}(t)\oplus x_{2}(t),~i\in[1,n]$. Then, we need $n$ observation nodes.
\end{remark}

Then we focus on $K$-XOR-BNs with arbitrary fixed $K>2$, and we first show that there exists a $K$-XOR-BN that is observable with $m=K$.

\begin{proposition}\label{proposition8}
There exists a $K$-XOR-BN that is observable with $m=K$, where $n=K+1$.
\end{proposition}
\begin{proof}
Consider the following BN:
\begin{eqnarray*}
x_{1}(t+1) & = & x_{2}(t)\oplus x_{3}(t)\oplus x_{4}(t)\oplus\cdots\oplus x_{K+1}(t),\\
x_{2}(t+1) & = & x_{1}(t)\oplus x_{3}(t)\oplus x_{4}(t)\oplus\cdots\oplus x_{K+1}(t),\\
x_{3}(t+1) & = & x_{1}(t)\oplus x_{2}(t)\oplus x_{4}(t)\oplus\cdots\oplus x_{K+1}(t),\\
&\cdots&\\
x_{K}(t+1) & = & x_{1}(t)\oplus x_{2}(t)\oplus\cdots\oplus x_{K-1}(t)\oplus x_{K+1}(t),\\
x_{K+1}(t+1) & = & x_{1}(t)\oplus x_{2}(t)\oplus x_{3}(t)\oplus\cdots\oplus x_{K}(t),
\end{eqnarray*}
where $K$ is odd.

First, we prove that if the number of ones to $\xvec(0)=[x_{1}(0),x_{2}(0),\ldots,x_{K+1}(0)]$ is even, then $\xvec(0)=\xvec(1)$. Suppose that $x_{j_{1}}(0)=x_{j_{2}}(0)=\cdots=x_{j_{2k}}(0)=1$. Then for any $i\in\{j_{1},j_{2},\ldots,j_{2k}\}$, we have $x_{i}(1)=1$ (the input has an odd number of ones); for any $i\in\{1,2,\ldots,K+1\}\setminus\{j_{1},j_{2},\ldots,j_{2k}\}$, we have $x_{i}(1)=0$ (the input has an even number of ones). Hence, if the number of ones to $\xvec(0)=[x_{1}(0),x_{2}(0),\ldots,x_{K+1}(0)]$ is even, then $\xvec(0)=\xvec(1)$, i.e., there are $\binom{K+1}{0}+\binom{K+1}{2}+\cdots+\binom{K+1}{K+1}=2^{K}$ fixed points. Here, any fixed point $\xvec(0)$ can only be determined from $\yvec(0)$, which implies $2^{m}\geq2^{K}$, i.e., $m\geq K$.

Then we consider that the number of ones to $\xvec(0)=[x_{1}(0),x_{2}(0),\ldots,x_{K+1}(0)]$ is odd. Suppose that $x_{j_{1}}(0)=x_{j_{2}}(0)=\cdots=x_{j_{2k+1}}(0)=1$. Then for any $i\in\{j_{1},j_{2},\ldots,j_{2k+1}\}$, we have $x_{i}(1)=0$ (the input has an even number of ones); for any $i\in\{1,2,\ldots,K+1\}\setminus\{j_{1},j_{2},\ldots,j_{2k}\}$, we have $x_{i}(1)=1$ (the input has an odd number of ones).

Hence, for any initial state $\xvec(0)=[x_{1}(0),x_{2}(0),\ldots,x_{K+1}(0)]$, we have $\xvec(0)=\xvec(2)=\cdots=\xvec(2k)$ and $\xvec(1)=\xvec(3)=\cdots=\xvec(2k+1)$. Then, we have the following state transition table (Table \ref{table13}).
\begin{table}[ht!]
  \centering
  \caption{State transition table in Proposition \ref{proposition8}.}
  \label{table13}
\resizebox{\linewidth}{!}{
\begin{tabular}{|ccccccc|ccccccc|}
\hline
$x_{1}(0)$&$x_{2}(0)$&$x_{3}(0)$&$\cdots$&$x_{K-1}(0)$&$x_{K}(0)$&$x_{K+1}(0)$&$x_{1}(1)$&$x_{2}(1)$&$x_{3}(1)$&$\cdots$&$x_{K-1}(1)$&$x_{K}(1)$&$x_{K+1}(1)$\\           \hline
$0$&$0$&$0$&$\cdots$&$0$&$0$&$0$&$0$&$0$&$0$&$\cdots$&$0$&$0$&$0$\\
$0$&$0$&$0$&$\cdots$&$0$&$0$&$1$&$1$&$1$&$1$&$\cdots$&$1$&$1$&$0$\\
$0$&$0$&$0$&$\cdots$&$0$&$1$&$0$&$1$&$1$&$1$&$\cdots$&$1$&$0$&$1$\\
$0$&$0$&$0$&$\cdots$&$0$&$1$&$1$&$0$&$0$&$0$&$\cdots$&$0$&$1$&$1$\\
$\vdots$&$\vdots$&$\vdots$&$\ddots$&$\vdots$&$\vdots$&$\vdots$&$\vdots$&$\vdots$&$\vdots$&$\ddots$&$\vdots$&$\vdots$&$\vdots$\\
$1$&$1$&$1$&$\cdots$&$1$&$0$&$0$&$0$&$0$&$0$&$\cdots$&$0$&$1$&$1$\\
$1$&$1$&$1$&$\cdots$&$1$&$0$&$1$&$1$&$1$&$1$&$\cdots$&$1$&$0$&$1$\\
$1$&$1$&$1$&$\cdots$&$1$&$1$&$0$&$1$&$1$&$1$&$\cdots$&$1$&$1$&$0$\\
$1$&$1$&$1$&$\cdots$&$1$&$1$&$1$&$0$&$0$&$0$&$\cdots$&$0$&$0$&$0$\\
\hline
\end{tabular}}
\end{table}

Clearly, $\xvec(0)$ can be recovered from $\yvec(0),\yvec(1)$, where
\begin{eqnarray*}
y_{1}(t) & = & x_{1}(t),\\
y_{2}(t) & = & x_{2}(t),\\
&\cdots&\\
y_{K}(t) & = & x_{K}(t).
\end{eqnarray*}
\end{proof}

\begin{example}\label{example}
Consider the following BN:
\begin{eqnarray*}
x_{1}(t+1) & = & x_{2}(t)\oplus x_{3}(t)\oplus x_{4}(t),\\
x_{2}(t+1) & = & x_{1}(t)\oplus x_{3}(t)\oplus x_{4}(t),\\
x_{3}(t+1) & = & x_{1}(t)\oplus x_{2}(t)\oplus x_{4}(t),\\
x_{4}(t+1) & = & x_{1}(t)\oplus x_{2}(t)\oplus x_{3}(t).
\end{eqnarray*}
Then, we have the following state transition table (Table \ref{table14}).
\begin{table}[ht!]
  \centering
  \caption{State transition table in Example \ref{example}.}
  \label{table14}
\begin{tabular}{|cccc|cccc|cccc|}
\hline
$x_{1}(0)$&$x_{2}(0)$&$x_{3}(0)$&$x_{4}(0)$&$x_{1}(1)$&$x_{2}(1)$&$x_{3}(1)$&$x_{4}(1)$&$x_{1}(2)$&$x_{2}(2)$&$x_{3}(2)$&$x_{4}(2)$\\           \hline
$0$&$0$&$0$&$0$&$0$&$0$&$0$&$0$&$0$&$0$&$0$&$0$\\
$0$&$0$&$0$&$1$&$1$&$1$&$1$&$0$&$0$&$0$&$0$&$1$\\
$0$&$0$&$1$&$0$&$1$&$1$&$0$&$1$&$0$&$0$&$1$&$0$\\
$0$&$0$&$1$&$1$&$0$&$0$&$1$&$1$&$0$&$0$&$1$&$1$\\
$0$&$1$&$0$&$0$&$1$&$0$&$1$&$1$&$0$&$1$&$0$&$0$\\
$0$&$1$&$0$&$1$&$0$&$1$&$0$&$1$&$0$&$1$&$0$&$1$\\
$0$&$1$&$1$&$0$&$0$&$1$&$1$&$0$&$0$&$1$&$1$&$0$\\
$0$&$1$&$1$&$1$&$1$&$0$&$0$&$0$&$0$&$1$&$1$&$1$\\
$1$&$0$&$0$&$0$&$0$&$1$&$1$&$1$&$1$&$0$&$0$&$0$\\
$1$&$0$&$0$&$1$&$1$&$0$&$0$&$1$&$1$&$0$&$0$&$1$\\
$1$&$0$&$1$&$0$&$1$&$0$&$1$&$0$&$1$&$0$&$1$&$0$\\
$1$&$0$&$1$&$1$&$0$&$1$&$0$&$0$&$1$&$0$&$1$&$1$\\
$1$&$1$&$0$&$0$&$1$&$1$&$0$&$0$&$1$&$1$&$0$&$0$\\
$1$&$1$&$0$&$1$&$0$&$0$&$1$&$0$&$1$&$1$&$0$&$1$\\
$1$&$1$&$1$&$0$&$0$&$0$&$0$&$1$&$1$&$1$&$1$&$0$\\
$1$&$1$&$1$&$1$&$1$&$1$&$1$&$1$&$1$&$1$&$1$&$1$\\
\hline
\end{tabular}
\end{table}
If this BN is observable, then any $\xvec(0)=[x_{1}(0),x_{2}(0),x_{3}(0),x_{4}(0)]$ can be determined from $\yvec(0),\yvec(1)$ (because $\yvec(0)=\yvec(2)=\cdots=\yvec(2k),\yvec(1)=\yvec(3)=\cdots=\yvec(2k+1)$), where $\yvec(t)=[y_{1}(t),\ldots,y_{m}(t)]$.

For  $\xvec(0)=[0,0,0,0],[0,0,1,1],\ldots,[1,1,1,1]$ (8 cases), $\yvec(0)=\yvec(1)$. Hence, if this BN is observable, $2^{m}\geq8$, i.e., the number of observation nodes in this BN must not be less than 3.
We can construct
\begin{eqnarray*}
y_{1}(t) & = & x_{1}(t),\\
y_{2}(t) & = & x_{2}(t),\\
y_{3}(t) & = & x_{3}(t),
\end{eqnarray*}
such that this BN is observable. Hence, the minimum number of observation nodes in this BN is 3.
\end{example}

Now according to Proposition \ref{proposition8} and Example \ref{example}, we present a new technique to study the lower bounds of the number of observation nodes $m$, which is applicable to any BNs, as shown below.

\begin{proposition}\label{proposition9}
For any BN, if the number of its fixed points is $l$, then $m\geq\log_{2}l$ must hold so that the BN is observable.
\end{proposition}
\begin{proof}
For any fixed point $\xvec(0)$, it can only be determined from $\yvec(0)$, which implies $2^{m}\geq l$, i.e., $m\geq\log_{2}l$.
\end{proof}

\begin{proposition}\label{proposition10}
There exists a $K$-XOR-BN that is observable with $m=\frac{K}{K+1}n$, where $K$ is odd and $(n\mod K+1)=0$.
\end{proposition}
\begin{proof}
It is enough to make $\frac{n}{K+1}$ copies of the BN shown in Proposition \ref{proposition8}.
\end{proof}

\begin{proposition}\label{proposition11}
There exists a $K$-XOR-BN with arbitrary fixed $K>2$ that is observable with $m=1$.
\end{proposition}
\begin{proof}
Consider the following $K$-XOR-BN:
\begin{eqnarray*}
x_{1}(t+1) & = & x_{1}(t)\oplus x_{2}(t)\oplus x_{3}(t)\oplus\cdots\oplus x_{K-3}(t)\oplus x_{K-2}(t)\oplus x_{K-1}(t)\oplus x_{K+1}(t)\\
x_{2}(t+1) & = & x_{2}(t)\oplus x_{3}(t)\oplus x_{4}(t)\oplus\cdots\oplus x_{K-2}(t)\oplus x_{K-1}(t)\oplus x_{K}(t)\oplus x_{K+1}(t)\\
x_{3}(t+1) & = & x_{1}(t)\oplus x_{3}(t)\oplus x_{4}(t)\oplus\cdots\oplus x_{K-2}(t)\oplus x_{K-1}(t)\oplus x_{K}(t)\oplus x_{K+1}(t)\\
& \cdots & \\
x_{K-1}(t+1) & = & x_{1}(t)\oplus x_{2}(t)\oplus x_{3}(t)\oplus\cdots\oplus x_{K-3}(t)\oplus x_{K-1}(t)\oplus x_{K}(t)\oplus x_{K+1}(t)\\
x_{K}(t+1) & = & x_{1}(t)\oplus x_{2}(t)\oplus x_{3}(t)\oplus\cdots\oplus x_{K-3}(t)\oplus x_{K-2}(t)\oplus x_{K}(t)\oplus x_{K+1}(t)\\
x_{K+1}(t+1) & = & x_{1}(t)\oplus x_{2}(t)\oplus x_{3}(t)\oplus\cdots\oplus x_{K-3}(t)\oplus x_{K-2}(t)\oplus x_{K-1}(t)\oplus x_{K}(t)\\
y(t) & = & x_{1}(t).
\end{eqnarray*}
Then when $K$ is even, we have
\begin{eqnarray*}
y(0) & = & x_{1}(0),\\
y(1) & = & x_{1}(0)\oplus x_{2}(0)\oplus\cdots\oplus x_{K-2}(0)\oplus x_{K-1}(0)\oplus x_{K+1}(0),\\
y(2) & = & x_{1}(0)\oplus x_{2}(0)\oplus\cdots\oplus x_{K-2}(0)\oplus x_{K}(0)\oplus x_{K+1}(0),\\
y(3) & = & x_{1}(0)\oplus x_{2}(0)\oplus\cdots\oplus x_{K-3}(0)\oplus x_{K-1}(0)\oplus x_{K}(0)\oplus x_{K+1}(0),\\
y(4) & = & x_{1}(0)\oplus x_{2}(0)\oplus\cdots\oplus x_{K-4}(0)\oplus x_{K-2}(0)\oplus x_{K-1}(0)\oplus x_{K}(0)\oplus x_{K+1}(0),\\
& \cdots & \\
y(K) & = & x_{2}(0)\oplus x_{3}(0)\oplus\cdots\oplus x_{K-2}(0)\oplus x_{K-1}(0)\oplus x_{K}(0)\oplus x_{K+1}(0),
\end{eqnarray*}
where $x_{1}(t)\oplus x_{2}(t)=x_{1}(t-1)\oplus x_{K}(t-1)$ and $x_{i}(t)\oplus x_{i+1}(t)=x_{i-1}(t-1)\oplus x_{i}(t-1),~i\in[3,K-1]$.

Suppose there are two different initial states $\xvec^{1}(0),\xvec^{2}(0)$, their corresponding outputs
\begin{eqnarray*}
[y^{1}(0),y^{1}(1),\ldots,y^{1}(K)]=[y^{2}(0),y^{2}(1),\ldots,y^{2}(K)].
\end{eqnarray*}
Here, $x^{1}_{1}(0)=x^{2}_{1}(0)$. According to $y^{1}(K-l+1)=y^{2}(K-l+1),~l\in[2:K]$, i.e., $x_{1}^{1}(0)\oplus\cdots\oplus x_{l-1}^{1}(0)\oplus x_{l+1}^{1}(0)\oplus\cdots\oplus x_{K+1}^{1}(0)=x_{1}^{2}(0)\oplus\cdots\oplus x_{l-1}^{2}(0)\oplus x_{l+1}^{2}(0)\oplus\cdots\oplus x_{K+1}^{2}(0)$ (there are an even (resp. odd) number of ones in $x_{1}^{1}(0),\ldots,x_{l-1}^{1}(0),x_{l+1}^{1}(0),\ldots,x_{K+1}^{1}(0)$ and an even (resp. odd) number of ones in $x_{1}^{2}(0),\ldots,x_{l-1}^{2}(0),x_{l+1}^{2}(0),\ldots,x_{K+1}^{2}(0)$) and $y^{1}(K)=y^{2}(K)$, we have $x^{1}_{l}(0)=x^{2}_{l}(0),~l\in[2:K]$. Based on $x^{1}_{l}(0)=x^{2}_{l}(0),~l\in[1:K]$ and $y^{1}(1)=y^{2}(1)$, we have $x^{1}_{K+1}(0)=x^{2}_{K+1}(0)$. In other words, if $[y^{1}(0),y^{1}(1),\ldots,y^{1}(K)]=[y^{2}(0),y^{2}(1),\ldots,y^{2}(K)]$, then $\xvec^{1}(0)=\xvec^{2}(0)$, which contradicts the hypothesis. Hence, $\xvec(0)$ can be recovered from $y(0),y(1),\ldots,y(K)$.

Similarly, when $K$ is odd, we have
\begin{eqnarray*}
y(0) & = & x_{1}(0),\\
y(1) & = & x_{1}(0)\oplus x_{2}(0)\oplus\cdots\oplus x_{K-2}(0)\oplus x_{K-1}(0)\oplus x_{K+1}(0),\\
y(2) & = & x_{K-1}(0),\\
y(3) & = & x_{1}(0)\oplus x_{2}(0)\oplus\cdots\oplus x_{K-3}(0)\oplus x_{K-1}(0)\oplus x_{K}(0)\oplus x_{K+1}(0),\\
y(4) & = & x_{K-3}(0),\\
& \cdots &\\
y(K-1) & = & x_{2}(0),\\
y(K) & = & x_{2}(0)\oplus x_{3}(0)\oplus\cdots\oplus x_{K-2}(0)\oplus x_{K-1}(0)\oplus x_{K}(0)\oplus x_{K+1}(0).
\end{eqnarray*}
Suppose there are two different initial states $\xvec^{1}(0),\xvec^{2}(0)$, their corresponding outputs
\begin{eqnarray*}
[y^{1}(0),y^{1}(1),\ldots,y^{1}(K)]=[y^{2}(0),y^{2}(1),\ldots,y^{2}(K)].
\end{eqnarray*}
It is clear that $x^{1}_{1}(0)=x^{2}_{1}(0)$ and $x^{1}_{2k}(0)=x^{2}_{2k}(0),~k\in[1,\frac{K-1}{2}]$. According to $y^{1}(K)=y^{2}(K)$ and $y^{1}(K-2k)=y^{2}(K-2k)$, we have $x^{1}_{2k+1}(0)=x^{2}_{2k+1}(0),~k\in[1,\frac{K-1}{2}]$. Based on $x^{1}_{l}(0)=x^{2}_{l}(0),~l\in[1:K]$ and $y^{1}(1)=y^{2}(1)$, we have $x^{1}_{K+1}(0)=x^{2}_{K+1}(0)$. In other words, if $[y^{1}(0),y^{1}(1),\ldots,y^{1}(K)]=[y^{2}(0),y^{2}(1),\ldots,y^{2}(K)]$, then $\xvec^{1}(0)=\xvec^{2}(0)$, which contradicts the hypothesis. Hence, $\xvec(0)$ can be recovered from $y(0),y(1),\ldots,y(K)$.
\end{proof}

\begin{example}\label{exmp1}
Consider the following BN:
\begin{eqnarray*}
x_{1}(t+1) & = & x_{1}(t)\oplus x_{2}(t)\oplus x_{4}(t),\\
x_{2}(t+1) & = & x_{2}(t)\oplus x_{3}(t)\oplus x_{4}(t),\\
x_{3}(t+1) & = & x_{1}(t)\oplus x_{3}(t)\oplus x_{4}(t),\\
x_{4}(t+1) & = & x_{1}(t)\oplus x_{2}(t)\oplus x_{3}(t),\\
y(t) & = & x_{1}(t).
\end{eqnarray*}
Then, we have the following state transition table (Table \ref{table15}).
\begin{table}[ht!]
  \centering
  \caption{State transition table in Example \ref{exmp1}.}
  \label{table15}
\resizebox{\linewidth}{!}{
\begin{tabular}{|cccc|cccc|cccc|cccc|}
\hline
$x_{1}(0)$&$x_{2}(0)$&$x_{3}(0)$&$x_{4}(0)$&$x_{1}(1)$&$x_{2}(1)$&$x_{3}(1)$&$x_{4}(1)$&$x_{1}(2)$&$x_{2}(2)$&$x_{3}(2)$&$x_{4}(2)$&$x_{1}(3)$&$x_{2}(3)$&$x_{3}(3)$&$x_{4}(3)$\\           \hline
$0$&$0$&$0$&$0$&$0$&$0$&$0$&$0$&$0$&$0$&$0$&$0$&$0$&$0$&$0$&$0$\\
$0$&$0$&$0$&$1$&$1$&$1$&$1$&$0$&$0$&$0$&$0$&$1$&$1$&$1$&$1$&$0$\\
$0$&$0$&$1$&$0$&$0$&$1$&$1$&$1$&$0$&$1$&$0$&$0$&$1$&$1$&$0$&$1$\\
$0$&$0$&$1$&$1$&$1$&$0$&$0$&$1$&$0$&$1$&$0$&$1$&$0$&$0$&$1$&$1$\\
$0$&$1$&$0$&$0$&$1$&$1$&$0$&$1$&$1$&$0$&$0$&$0$&$1$&$0$&$1$&$1$\\
$0$&$1$&$0$&$1$&$0$&$0$&$1$&$1$&$1$&$0$&$0$&$1$&$0$&$1$&$0$&$1$\\
$0$&$1$&$1$&$0$&$1$&$0$&$1$&$0$&$1$&$1$&$0$&$0$&$0$&$1$&$1$&$0$\\
$0$&$1$&$1$&$1$&$0$&$1$&$0$&$0$&$1$&$1$&$0$&$1$&$1$&$0$&$0$&$0$\\
$1$&$0$&$0$&$0$&$1$&$0$&$1$&$1$&$0$&$0$&$1$&$0$&$0$&$1$&$1$&$1$\\
$1$&$0$&$0$&$1$&$0$&$1$&$0$&$1$&$0$&$0$&$1$&$1$&$1$&$0$&$0$&$1$\\
$1$&$0$&$1$&$0$&$1$&$1$&$0$&$0$&$0$&$1$&$1$&$0$&$1$&$0$&$1$&$0$\\
$1$&$0$&$1$&$1$&$0$&$0$&$1$&$0$&$0$&$1$&$1$&$1$&$0$&$1$&$0$&$0$\\
$1$&$1$&$0$&$0$&$0$&$1$&$1$&$0$&$1$&$0$&$1$&$0$&$1$&$1$&$0$&$0$\\
$1$&$1$&$0$&$1$&$1$&$0$&$0$&$0$&$1$&$0$&$1$&$1$&$0$&$0$&$1$&$0$\\
$1$&$1$&$1$&$0$&$0$&$0$&$0$&$1$&$1$&$1$&$1$&$0$&$0$&$0$&$0$&$1$\\
$1$&$1$&$1$&$1$&$1$&$1$&$1$&$1$&$1$&$1$&$1$&$1$&$1$&$1$&$1$&$1$\\
\hline
\end{tabular}}
\end{table}
Clearly, $\xvec(0)$ can be recovered from $y(0),y(1),y(2),y(3)$.
\end{example}

\section{Boolean networks consisting of nested canalyzing functions}
In this section, we study the number of observation nodes in BNs consisting of nested canalyzing functions, and we refer to these BNs, where each Boolean function is a nested canalyzing function of $K$ literals, as $K$-NC-BNs.

\begin{theorem}\label{theorem5}
For any $n$ and $K>2$,
there exists a BN satisfying
(i) for each node, a nested canalyzing function
with $K$ input nodes is assigned, and
(ii) the number of observation nodes is $\lceil \frac{n}{K} \rceil$.
\end{theorem}
\begin{proof}
First, we consider the case of $K=n$, where $n>2$.
In this case, we construct the BN defined by
\begin{eqnarray*}
x_1(t+1) & = & x_2(t) \lor (\lnon{x_1(t)} \land \lnon{x_3(t)} \land \lnon{x_4(t)} \land \cdots \land \lnon{x_n(t)}),\\
x_2(t+1) & = & x_3(t) \lor (\lnon{x_1(t)} \land \lnon{x_2(t)} \land \lnon{x_4(t)} \land \cdots \land \lnon{x_n(t)}),\\
& & \cdots \\
x_{n-1}(t+1) & = & x_{n}(t) \lor (\lnon{x_1(t)} \land \lnon{x_2(t)} \land \cdots \land \lnon{x_{n-2}(t)} \land \lnon{x_{n-1}(t)}),\\
x_{n}(t+1) & = & x_{1}(t) \lor (\lnon{x_2(t)} \land \lnon{x_3(t)} \land \lnon{x_4(t)} \land \cdots \land \lnon{x_n(t)}).
\end{eqnarray*}
Let $y_1(t) = x_1(t)$ be the observation node.
Then, it is straightforward to see:
$$
\xvec(0)=[y_1(0),y_1(1),\ldots,y_1(n-1)]
$$
unless $\xvec(0)=[0,0,\ldots,0]$.

For the case of $\xvec(0)=[0,0,\ldots,0]$, we have
$$
[y_1(0),y_1(1),\ldots,y_1(n-1),y_1(n)]=[0,1,\ldots,1,1].
$$
On the other hand, if $\xvec(0)=[0,1,\ldots,1]$, we have
$$
[y_1(0),y_1(1),\ldots,y_1(n-1),y_1(n)]=[0,1,\ldots,1,0].
$$
Therefore, we can discriminate $\xvec(0)=[0,0,\ldots,0]$
from all other initial states.

Next, we consider the case of $(n \mod K)=0$.
In this case, it is enough to partition $[x_1,x_2,\ldots,x_n]$
into $\frac{n}{K}$ blocks: $[x_1,x_2,\ldots,x_K],[x_{K+1},x_{K+2},\ldots,x_{2K}],\ldots$.
We construct a partial BN for each block as in the first case,
then merge these BNs,
and let
$$
[y_1(t),y_2(t),\ldots,y_{\frac{n}{K}}(t)]=[x_1(t),x_{K+1}(t),\ldots,x_{n-K+1}(t)].
$$
Clearly, this BN is observable.

Finally, we consider the case of $(n \mod K) \neq 0$.
As in the second case, we partition
$[x_1,x_2,\ldots,x_n]$
into $\lceil \frac{n}{K} \rceil$ blocks and
construct a partial BN for each block as in the first case
except for the last block.
Let $x_h,\ldots,x_n$ be the variables in the last block,
where $h = K (\lceil \frac{n}{K} \rceil-1)+1$.
Let $p=n-h+1$ and $q=K-p$.
Note that the last block contains $p$ variables.
Then, we construct a partial BN for the last block as below.
\begin{eqnarray*}
x_{h}(t+1) & = & x_{h+1}(t) \lor (\lnon{x_{h}(t)} \land \lnon{x_{h+2}(t)} \land \cdots \land \lnon{x_n(t)} \land
x_1(t) \land x_2(t) \land \cdots \land x_{q}(t)),\\
x_{h+1}(t+1) & = & x_{h+2}(t) \lor (\lnon{x_{h}(t)} \land \lnon{x_{h+1}(t)} \land \cdots \land \lnon{x_n(t)} \land
x_1(t) \land x_2(t) \land \cdots \land x_{q}(t)),\\
& & \cdots\\
x_{n}(t+1) & = & x_{h}(t) \lor (\lnon{x_{h+1}(t)} \land \lnon{x_{h+2}(t)} \land \cdots \land \lnon{x_n(t)} \land
x_1(t) \land x_2(t) \land \cdots \land x_{q}(t)).
\end{eqnarray*}
We define the observation nodes by
$$
[y_1(t),y_2(t),\ldots,y_r(t)]
=
[x_1(t),x_{K+1}(t),\ldots,x_{h}(t)],$$
where $r = \lceil \frac{n}{K} \rceil$.

It is straightforward to see that
if $[x_h(0),x_{h+1}(0),\ldots,x_n(0)] \neq [0,0,\ldots,0]$,
we have
\begin{eqnarray*}
& & [y_{r}(0),y_{r}(1),\ldots,y_{r}(p-1),y_{r}(p),y_{r}(p+1),\ldots,y_{r}(2p-1),\ldots]\\
&=&
[x_{h}(0),x_{h+1}(0),\ldots,x_{n}(0),x_{h}(0),x_{h+1}(0),\ldots,x_{n}(0),\ldots].
\end{eqnarray*}
If $[x_h(0),x_{h+1}(0),\ldots,x_n(0)] = [0,0,\ldots,0]$,
then $y_{r}(0)=0$ holds, and either $y_{r}(t)=0$ holds for all $t>0$,
or $y_{r}(t)=1$ holds for all $t > t'$ for some $t'$.
Thus, all distinct $[x_h(0),\ldots,x_n(0)]$'s can be discriminated
by observing $y_r(t)$.
Therefore, this BN is also observable.
\end{proof}

\begin{example}
Consider the case of $n=11$ and $K=4$.
In this case, $h=9$, $p=3$, $q=1$, and $r=3$.
The BN for the first block is constructed by
\begin{eqnarray*}
x_{1}(t+1) & = & x_{2}(t) \lor (\lnon{x_{1}(t)} \land \lnon{x_{3}(t)} \land
\lnon{x_{4}(t)}),\\
x_{2}(t+1) & = & x_{3}(t) \lor (\lnon{x_{1}(t)} \land \lnon{x_{2}(t)} \land
\lnon{x_{4}(t)}),\\
x_{3}(t+1) & = & x_{4}(t) \lor (\lnon{x_{1}(t)} \land \lnon{x_{2}(t)} \land
\lnon{x_{3}(t)}),\\
x_{4}(t+1) & = & x_{1}(t) \lor (\lnon{x_{2}(t)} \land \lnon{x_{3}(t)} \land
\lnon{x_{4}(t)}).
\end{eqnarray*}
The BN for the second block is constructed by
\begin{eqnarray*}
x_{5}(t+1) & = & x_{6}(t) \lor (\lnon{x_{5}(t)} \land \lnon{x_{7}(t)} \land
\lnon{x_{8}(t)}),\\
x_{6}(t+1) & = & x_{7}(t) \lor (\lnon{x_{5}(t)} \land \lnon{x_{6}(t)} \land
\lnon{x_{8}(t)}),\\
x_{7}(t+1) & = & x_{8}(t) \lor (\lnon{x_{5}(t)} \land \lnon{x_{6}(t)} \land
\lnon{x_{7}(t)}),\\
x_{8}(t+1) & = & x_{5}(t) \lor (\lnon{x_{6}(t)} \land \lnon{x_{7}(t)} \land
\lnon{x_{8}(t)}).
\end{eqnarray*}
The BN for the last block is constructed by
\begin{eqnarray*}
x_{9}(t+1) & = & x_{10}(t) \lor (\lnon{x_{9}(t)} \land \lnon{x_{11}(t)} \land
x_1(t)),\\
x_{10}(t+1) & = & x_{11}(t) \lor (\lnon{x_{9}(t)} \land \lnon{x_{10}(t)} \land
x_1(t)),\\
x_{11}(t+1) & = & x_{9}(t) \lor (\lnon{x_{10}(t)} \land \lnon{x_{11}(t)} \land
x_1(t)).
\end{eqnarray*}
And, the observation nodes are defined by
$[y_1(t),y_2(t),y_3(t)]=[x_1(t),x_5(t),x_9(t)]$.
\end{example}

\begin{remark}
Note that the upper bound of the number of observation nodes for $K$-NC-BNs obtained in Theorem \ref{theorem5} can break the lower bound of the number of observation nodes for $K$-AND-OR-BNs obtained in Theorem \ref{theorem2}.
\end{remark}

\begin{proposition}\label{proposition}
For any $K$-NC-BN, $m\geq(1-\beta_{K})n$ must hold so that the BN is observable, where $\beta_{K}=-\left[\frac{2^{K-1}-1}{2^K}\log_{2}\left(\frac{2^{K-1}-1}{2^{K}}\right)+\frac{2^{K-1}+1}{2^K}\log_{2}\left(\frac{2^{K-1}+1}{2^{K}}\right)\right]$.
\end{proposition}
\begin{proof}
It is seen from Theorem 23 of \cite{guo22} that for any canalyzing function of $K$ literals, the entropy of $x_{1}(1)$ is maximized when the $x_{1}(1)=0$ and $x_{1}(1)=1$ for $2^{K-1}-1$ and $2^{K-1}+1$ inputs, respectively. In other words,
the entropy of $x_{1}(1)$ is no greater than $\beta_{K}$, where
\begin{eqnarray*}
\beta_{K}=-\left[\frac{2^{K-1}-1}{2^K}\log_{2}\left(\frac{2^{K-1}-1}{2^{K}}\right)+\frac{2^{K-1}+1}{2^K}\log_{2}\left(\frac{2^{K-1}+1}{2^{K}}\right)\right].
\end{eqnarray*}

In total, the information
quantify of $\xvec(1)$ is at most $n\beta_{K}$ bits, which might be obtained by $[\yvec(0),\ldots,\yvec(N)]$. However, in order to observe any initial state $\xvec(0)$, we should have the information quantity of $n$ bits. Therefore, at least $n-n\beta_{K}=(1-\beta_{K})n$ bits must be provided from $\yvec(0)$, which implies $m\geq(1-\beta_{K})n$.
\end{proof}

\begin{remark}
Comparing the coefficient $\frac{1}{K}$ in Theorem \ref{theorem5} with the coefficient $1-\beta_{K}$ in Proposition \ref{proposition}, we notice that $\lceil \frac{n}{K} \rceil>(1-\beta_{K})n$ always holds regardless of $n$, and we illustrate it as follows (Fig~\ref{2}). It is seen from Theorem 23 of \cite{guo22} that for sufficiently large $K$, $\beta_{K}\approx1-\frac{1}{\ln2}\left(\frac{1}{2^{2K-1}}\right)$.
Hence, $\frac{1}{K}>1-\beta_{K}$, which is reasonable from the definitions
of upper and lower bounds.
However, as seen from Fig~\ref{2},
the gap between $\lceil \frac{n}{K} \rceil$ and $(1-\beta_{K})n$ is not small,
in particular compared with that on Fig~\ref{1},
Hence, further studies should be done for closing the gap.
It should also be noted that the curves in Fig.~\ref{1} are increasing
whereas those in Fig.~\ref{2} are decreasing.
However, there is no inconsistency and this observation suggest that
there exist big differences on the observability between
$K$-AND-OR-BNs and $K$-NC-BNs.
\begin{figure}[htb]
\centering
\includegraphics[width=0.6\textwidth]{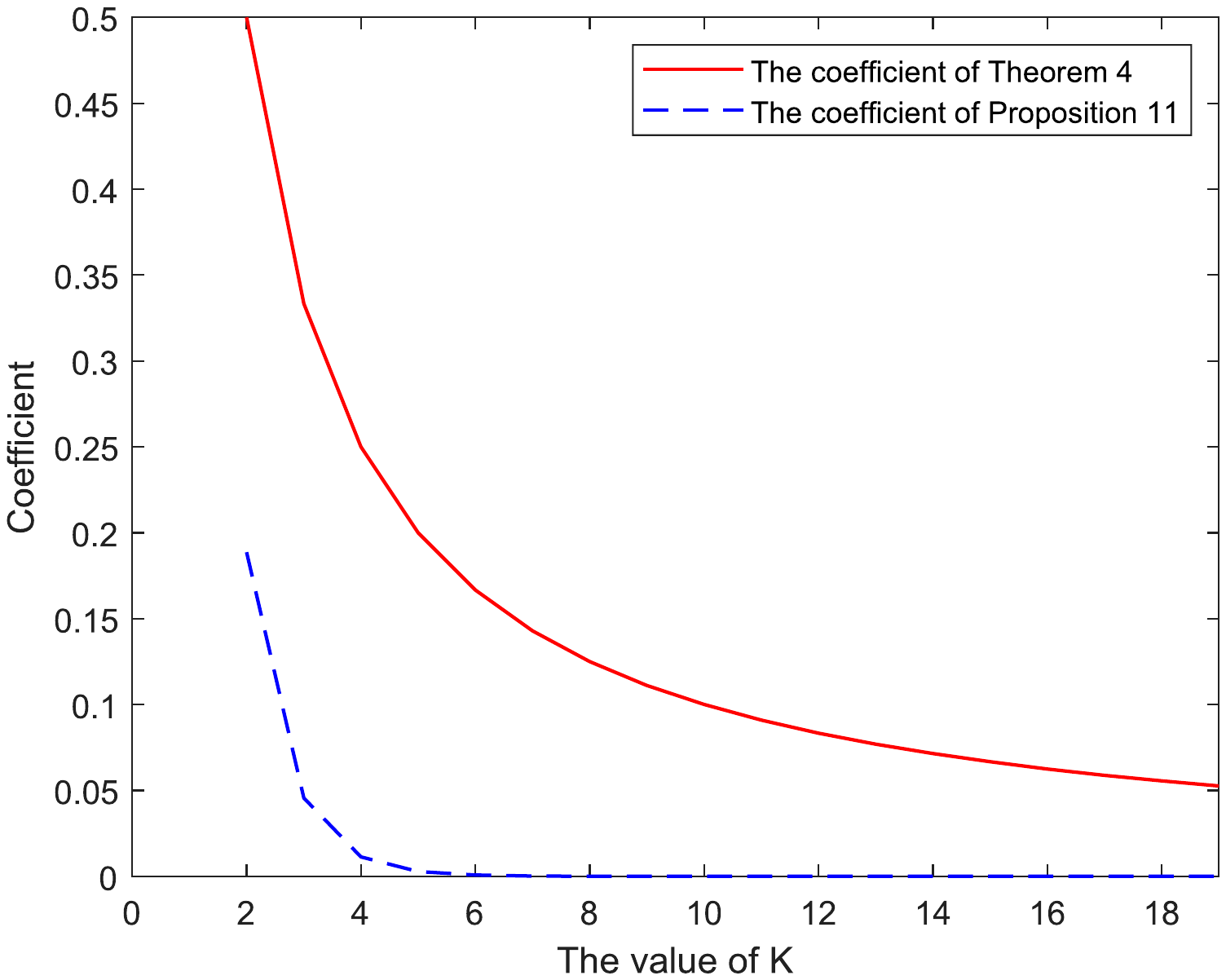}
\caption{Comparison between the coefficient $\frac{1}{K}$ in Theorem \ref{theorem5} with the coefficient $1-\beta_{K}$ in Proposition \ref{proposition}. Here, we find that $\frac{1}{K}>1-\beta_{K}$.}
\label{2}
\end{figure}
\end{remark}
\section{Conclusion}
The number of observation nodes plays a crucial role in the observability of BNs. In this paper, we mainly studied the upper and lower bounds of the number of observation nodes for three types of BNs. In other words, we analyzed and obtained the general lower bound, best case upper bound, worst case lower bound, and general upper bound for $K$-AND-OR-BNs, $K$-XOR-BNs, and $K$-NC-BNs, respectively.
First, for $K$-AND-OR-BNs, a novel technique using information entropy to derive a lower bound of the number of observation nodes was proposed, where $\left[(1-K)+\frac{2^{K}-1}{2^{K}}\log_{2}(2^{K}-1)\right]n$ was inferred as a general lower bound. Meanwhile, the best case upper bound $\left(\frac{2^{K}-K-1}{2^{K}-1}\right)n$ for $K$-AND-OR-BNs was found. Then, for $K$-XOR-BNs, we found that there exists a $K$-XOR-BN that is observable with $m=K$, and that there exists a $K$-XOR-BN that is observable when $m=\frac{K}{K+1}n$. In addition, the best case upper bound for $K$-XOR-BNs was $1$. Finally, we showed that the general lower bound for $K$-NC-BNs is $(1-\beta_{K})n$ and that there exists a specific $K$-NC-BN, which is observable with $m=\lceil \frac{n}{K}\rceil$.
Here, we found that the upper bound for $K$-NC-BNs could break the lower bound for $K$-AND-OR-BNs, which implies that some $K$-NC-BNs are easier to observe than any $K$-AND-OR-BNs.

For any BNs, we also developed two new techniques for inferring the general lower bounds of the number of observation nodes, namely using counting identical states $\xvec(1;\xvec(0))$ and counting the number of fixed points. Generally speaking, by counting identical states $\xvec(1;\xvec(0))$, a general lower bound $\log_{2}[\max_{\avec^{j}}COUNT(\avec^{j})]$, where $COUNT(\avec^{j})$ represents the total number of occurrences of $\avec^{i}$ in $\{\xvec(1;\xvec(0)),\xvec(0)\in\{0,1\}^{n}\}$, was derived; by counting the number of fixed points, a general lower bound $\log_{2}l$, where $l$ represents the number of fixed points in a BN, was derived.

In summary, through combinatorial analysis of several types of BNs, we derived nontrivial lower and upper bounds of the number of observation nodes.
However, there exist gaps between the best case upper bounds and the general lower bounds.
In particular, the gap between Theorem \ref{theorem5} and
Proposition \ref{proposition} is not small.
Therefore, closing gaps is left as interesting open problems.

\end{document}